\documentclass[letterpaper, 10 pt, conference]{ieeeconf}  % Comment this line out
\IEEEoverridecommandlockouts                              % This command is only
\usepackage{graphicx} % for pdf, bitmapped graphics files
\usepackage{epsfig} % for postscript graphics files
\usepackage{mathptmx} % assumes new font selection scheme installed 
\usepackage{times} % assumes new font selection scheme installed
\usepackage{bm}
\usepackage{mathtools}
 \usepackage{nccmath}
\usepackage[font=footnotesize]{caption}
\usepackage[font=footnotesize]{subcaption}
\usepackage{float}
\usepackage{optidef}
\usepackage{amsmath} % assumes amsmath package installed
\usepackage{bm}
\usepackage{mathtools}
 \usepackage{nccmath}
\usepackage{amssymb}  % assumes amsmath package installed
\usepackage[ruled,vlined,linesnumbered]{algorithm2e}
\usepackage{tikz}
\usepackage{amsfonts}
\usetikzlibrary{shapes,arrows,positioning,fit}

\DeclareMathOperator*{\argmin}{argmin}
\graphicspath{{./figs/}}
\usepackage{color}
\newtheorem{remark}{Remark}
\newtheorem{theorem}{Theorem}
\newtheorem{lemma}{Lemma}
\newtheorem{problem}{Problem}
\newtheorem{proposition}{Proposition}

\newtheorem{definition}{Definition}

\renewcommand{\baselinestretch}{0.996}
\newcommand{\real}{\ensuremath{\mathbb{R}}}
\newcommand{\realpositive}{\ensuremath{\mathbb{R}}_{>0}}
\newcommand{\realnonnegative}{\ensuremath{\mathbb{R}}_{\ge 0}}

\newcommand{\until}[1]{\{1,\dots, #1\}}
\newcommand{\subscr}[2]{#1_{\textup{#2}}}
\newcommand{\supscr}[2]{#1^{\textup{#2}}}

\newcommand{\longthmtitle}[1]{\mbox{}{\text{(#1).}}}

\newcommand\oprocendsymbol{\hbox{$\bullet$}}
\newcommand\oprocend{\relax\ifmmode\else\unskip\hfill\fi\oprocendsymbol}

\title{\LARGE \bf
Risk-perception-aware control design under dynamic spatial risks
}

\author{Aamodh Suresh and Sonia Mart{\'\i}nez% <-this % stops a space
  \thanks{A.~Suresh and S.~Mart{\'\i}nez are with the Department of
  	Mechanical and Aerospace Engineering, University of California at
  	San Diego, La Jolla, CA 92093, USA {\tt\small
  		\{aasuresh,soniamd\}@eng.ucsd.edu}}%
  \thanks{We gratefully acknowledge support from ONR grant ONR - N00014-19-1-2471. }  }

\begin{document}

\tikzstyle{block} = [rectangle, draw, fill=blue!20, 
text width=5em, text centered, rounded corners, minimum height=2.5em]
\tikzstyle{bigblock} = [rectangle, draw, fill=blue!20, 
text width=7.5em, text centered, rounded corners, minimum height=2.5em]    
\tikzstyle{line} = [draw, -latex']
\tikzstyle{cloud} = [draw, ellipse,fill=red!20, node distance=3cm,
minimum height=2em]

\maketitle
\thispagestyle{empty}
\pagestyle{empty}

%%%%%%%%%%%%%%%%%%%%%%%%%%%%%%%%%%%%%%%%%%%%%%%%%%%%%%%%%%%%%%%%%%%%%%%%%%%%%%%%
\begin{abstract}
 This work proposes a novel risk-perception-aware
 (RPA) control design using non-rational perception of risks
 associated with uncertain dynamic spatial costs. We use Cumulative
 Prospect Theory (CPT) to model the risk perception
 of a decision maker (DM) and use it to construct perceived
 risk functions that transform the uncertain dynamic spatial
 cost to deterministic perceived risks of a DM. These risks
 are then used to build safety sets which can represent risk-averse
 to risk-insensitive perception. We define a notions of
 “inclusiveness” and “versatility” based on safety sets and use
 it to compare with other models such as Conditional value at
 Risk (CVaR) and Expected risk (ER). We theoretically prove
 that CPT is the most “inclusive” and “versatile” model of the
 lot in the context of risk-perception-aware controls. We further
 use the perceived risk function along with ideas from control
 barrier functions (CBF) to construct a class of perceived risk
 CBFs. For a class of truncated-Gaussian costs, we find sufficient
 geometric conditions for the validity of this class of CBFs, thus
 guaranteeing safety. Then, we generate perceived-safety-critical
 controls using a Quadratic program (QP) to guide an agent
 safely according to a given perceived risk model. We present
 simulations in a 2D environment to illustrate the performance
 of the proposed controller.
\end{abstract}

%%%%%%%%%%%%%%%%%%%%%%%%%%%%%%%%%%%%%%%%%%%%%%%%%%%%%%%%%%%%%%%%%%%%%%%%%%%%%%%%
\section{INTRODUCTION}
\label{sec:introduction}

\emph{Motivation:} Safety is a desirable and necessary design
constraint for any control system; specially when operated in a shared
environment with a decision maker (DM). 
Arguably, most environments have associated spatial risks, whose
source can vary from hard constraints (e.g.~moving obstacles) to softer
constraints (e.g.~wind conditions). Different DMs can perceive these
risks differently, leading to notions of perceived risks and perceived
safety from these risks.

It is well known from psychophysics~\cite{SSS:70} and behavioral
economics~\cite{AT-DK:92} research that humans as DMs have
fundamental non-linear perception leading to non-rational decision
making in risky situations. 
In such cases, existing methods assuming perfect knowledge or
rational and coherent treatment (as in expected risk and Conditional Value at Risk
(CVaR)) of risks may not suffice, which can lead to loss of trust or
discomfort among DMs.
This motivates the need of richer and more inclusive modeling of risk
perception to capture a variety of DMs and use them for safe control
design. This work aims to bridge the gap between behavioral decision
making and safety using Cumulative Prospect Theory (CPT) as a risk
perception model, and Control Barrier Functions (CBFs) for safe
control design.

\emph{Related Work:} Safe control system design has been tackled using
various frameworks such as artificial potential
functions~\cite{OK:90}, barrier certificates~\cite{SP-AJ:04} and, more
recently, control barrier
functions (CBFs)~\cite{ADA-SC-ME-GN-KS-PT:19}. %\cite{PW-FA:07,ADA-SC-ME-GN-KS-PT:19}.
% The latter two trace their roots to the works of Nagumo~\cite{MN:42}
% based on set invariance principles.
CBFs have gained popularity due to their Lyapunov-like properties,
rigorous safety guarantees and ease of
application. %~\cite{AA-KS:17,RC-GO-RMM-JWB:19}.
They have been successfully used in
optimization~\cite{ADA-SC-ME-GN-KS-PT:19},
stabilization~\cite{PO-JC:19-cdc} and data-driven control
frameworks~\cite{BTL-JJS-JPH:21}. CBFs were traditionally used in
static scenarios, more recently, they have been used to deal with
moving obstacles~\cite{YC-HP-JG:18} and multi-agent
systems~\cite{PG-JC-ME:21-tac}.
%{\color{red} Cite a few recent CBF applications}.

Uncertainty has been mainly handled using robustness
measures~\cite{SP-AJ-GJP:07}, stochastic control~\cite{AC:19}, or
chance constraints~\cite{MJK-VD-MF-NA:20}. Very few works have
considered the notion of risk perception explicitly in a control
system~\cite{SS-IY:18,MA-XX-ADA:21}. %\cite{SS-IY:18,LL-GJP-DVD:20,MA-XX-ADA:20}.
All these works use CVaR %~\cite{RTR-SU:00}
to quantify risk perception, which only captures linear and rational
risk-averse behavior. CPT on the other hand is a more expressive
(see~\cite{AS-SM:21-ral}), non-linear and non-rational perception
theory which is yet to be applied in the context of safety for a
control system. Moreover, CPT has been successfully used in
engineering applications like path planning~\cite{AS-SM:21-ral},
traffic routing~\cite{SG-EF-MBA:10}, and network
protection~\cite{ARH-SS:19}. % , and stochastic
% optimization~\cite{CJ-PLA-MF-SM-CS:18} to model non-rational decision
% making.

\emph{Contributions:} We first adapt the notion of non-rational risk
perception to the context of safety for control systems. With this, we
capture a larger spectrum of DM's risk profile, extending the existing
literature. We support this claim theoretically by defining the notion
of ``inclusiveness'' and proving that CPT is the most inclusive risk
perception model out of the other popular models: CVaR and ER.
We then use the CPT value function to construct a class of CBFs to
guarantee safety according to a DM's perceived risk and define the
notion of perceived safety. Additionally, we find sufficient geometric
conditions on the control input to maintain the validity of our
proposed RPA CBF and compare them among the three risk perception models (RPMs). Then, we design a QP-based RPA controller to guide
an agent to a desired goal safely w.r.t.~perceived risks. Thus we
extend the literature with more inclusive safe control
design. Practically, we consider 2D simulations with moving obstacles
and show the effectiveness of the proposed RPA controller along with
the practical translation of the inclusiveness heirarchy.

This work % opens new doors in trying to infuse behavioral decision
% theory into rigorous control theory. It
provides a framework to %which can
incorporate and compare a wide range of RPMs
to generate a variety of RPA controls. We would also like to clarify
that the validation of CPT models using user studies for typical
control scenarios is beyond the scope of this work.

\section{Risk perception formalism and problem setup}
\label{sec:rpm}
Here, we introduce some notation\footnote{The
	Euclidean norm in $\real^n$ is denoted by $\|.\|$. % , and $\circ$ denotes
	% function composition; i.e.~$f(g(x))=f \circ g(x)$
	We use $\mathbb{E}$ as
	the expectation operator on a random variable. The set
	$\mathbf{B}^{r}(y) \triangleq \{x \in \mathcal{X}| \|x-y\| \leq r\}$
	is a ball of radius $r$ centered at $y$.} and a formal notion of
risk perception, starting with a concise description of
CPT and CVaR (see~\cite{SD:16} and~\cite{RTR-SU:00} for more
details). Later, we describe our problem statement.

\emph{Risk Perception:} By risk perception, we refer to the notion of
attaching a value (risk) to a random cost output.  Formally, let
$\mathcal{S}$ be a discrete sample space endowed with a probability
distribution $\mathbb{P}$. We model environmental cost via a
real-valued, discrete random variable $c:\mathcal{S} \rightarrow
\realnonnegative$, taking $M$ possible values, $c_i \in
\realnonnegative$, $i \in \until{M}$, and such that $p_i =
\mathbb{P}(c = c_i)$, with $\sum_i^M p_i=1$. We Let $\mathcal{C}$ be
the set of such random cost variables and $R:\mathcal{C} \rightarrow
\realnonnegative $ a value function which associates a value (risk) to
a random cost variable.

A value function $R$ can be defined in many ways, resulting in
different risk perceptions. Here, an RPM is characterized as a
parameterized %(by $\Theta$)
family $\mathcal{M} \triangleq \{R_{\Theta}| \Theta \in \real^l\}$ of
value functions.
% The set
% $\mathcal{R}$ captures the possible risk values that an RPM
% $\mathcal{M}$ can produce using every value function $R_\Theta \in
% \mathcal{M}$ for a given random variable $c$.
In what follows, we consider three popular RPMs: Expected Risk
(ER) % \footnote{The ER model reduces to a single value function defined
% via the expectation operator $\supscr{R}{e}(c)=\mathbb{E}(c)=\sum_i c_ip_i$.},
Conditional Value at Risk\footnote{The CVaR model uses a class of
	value functions parameterized by $q \in [0,1]$ to represent expectation over a fraction ($q$) of the worst-case
	outcomes.
	Thus the CVaR value with $q = 1$ 
	is the worst-case outcome of $c$, $c_M$.
	While, with $q = 0$ 
	CVaR value equals ER ($\supscr{R}{CV}_0=\supscr{R}{e}$).}
(CVaR)~\cite{RTR-SU:00}
and Cumulative Prospect Theory (CPT)~\cite{AT-DK:92}.

CPT captures
non-rational decision making, and was introduced
in~\cite{SD:16,DP:98}. In CPT, outcomes are first weighed
using a non-linear utility function $v: \realnonnegative \rightarrow
\realnonnegative$, with $v(c)=\lambda c^\gamma$,
modeling a DM's perceived cost. The parameters $\lambda \in
[1,\infty), \gamma \in [0,1]$ represent ``risk aversion'' and ``risk
sensitivity'', respectively.  In addition, a non-linear probability
weighing function $w :[0,1] \rightarrow [0,1]$, given by
$w(p)=e^{-\beta(-\log p)^{\alpha}}$ and $w(0)=0$, is used to model
uncertainty perception. Here, uncertainty sensitivity is  tuned via the
parameters $\alpha,\beta \in
\realpositive$. % help tune the uncertainty sensitivity of the DM.
CPT also suggests that probabilities are perceived via decision
weights $\Pi_i \in [0,1]$, which are calculated in a cumulative
fashion. Defining a partial sum function as $S_i(M) \triangleq
\sum_{j=i}^M p_j$, $\forall j \in \{1,...,M\}$ and $S_0(M) \triangleq
0$, we have $\Pi_j =w\circ S_j(M) - w \circ S_{j-1}(M). $

With this, assigning the parameter $q$ for CVaR and
$\theta=\{\alpha,\beta,\gamma,\lambda \}$ for CPT, the value functions
of ER ($\supscr{R}{ER}$), CVaR ($\supscr{R}{CV}$) and CPT ($\supscr{R}{cpt}$) of a DM are defined as:
\begin{subequations}\label{eqn:risk_def}
	\begin{align}
		\supscr{R}{ER}(c)\triangleq&\ \mathbb{E}(c)=\sum_{i=1}^M c_ip_i, \label{eqn:exp_risk}\\
		\supscr{R}{CV}_q(c)\triangleq&\ \mathbb{E}\left[c|c\geq \min \{d:\mathbb{P}(c \leq d)\geq q\} \right],\label{eqn:cvar_risk}\\
		\supscr{R}{cpt}_\theta(c)\triangleq&\ \sum_{j=1}^M (v \circ c_j) \Pi_j .
		\label{eqn:cpt_risk}
	\end{align}
\end{subequations} 
In CPT, $\theta$ can be varied to generate different value functions
pertaining to various risk profiles of DMs (from
risk-taker to risk-averse). We refer
to~\cite{AS-SM:21-ral,SD:16} for more details on the parameter
choices in
CPT.
\textit{Risky Environment:} Consider a compact state space
$\mathcal{X} \subset \real^n$ containing dynamic spatial sources of
risk at $y \in \mathcal{X}$ and an agent or robot at a state $x \in
\mathcal{X}$. The relative state space is $\mathcal{Z}\triangleq
\{\xi = y-x| x \in \mathcal{X}, \, y \in \mathcal{X}\}$. 
Our starting point is an uncertain cost field
$c: \mathcal{Z} \rightarrow \realnonnegative$,
%\margin{According to the
%  previous section, let's change $\mathcal{R}$ by
%  $\realnonnegative^M$} 
that aims to quantify objectively the (negative) consequences of being
at $x \in \mathcal{X}$ relative to a known risk source at $y \in
\mathcal{X}$. More precisely, $c(\xi)$ is a discrete RV
which can take $M$ possible values, $c_i(\xi) \in \realnonnegative$,
for $i \in \until{M}$. 
We assume that $c$ has associated mean and standard
deviation functions $c_\mu: \mathcal{Z} \rightarrow
\realnonnegative$ and $c_\sigma: \mathcal{Z} \rightarrow
\realnonnegative$, respectively. % ; that is, $c_\mu(\xi)$ defines the
% mean, and $c_\sigma(\xi)$ defines the standard deviation of the
% random variable $c(\xi)$, for each $\xi \in \mathcal{Z}$. 
We assume
that $c_\mu,c_\sigma$ are continuously differentiable in their
domains.
Given $c$, an associated spatial-risk function is given by
$R_c:\mathcal{Z} \rightarrow \realnonnegative$, $R_c(\xi) \equiv
R(c(\xi))$, where $R$ belongs to any of the previous RPMs defined in
\eqref{eqn:risk_def} above. When clear from the context, we will
identify $R_c \equiv R \in \mathcal{M}$.
%\margin{if we are fixing the cost $c$, then
%  maybe we should denote as $R_c$ and not just $R$. WE'll see if this
%  becomes confusing or not} 
The larger $R_c$ is at $\xi$, the higher the perceived risk of being
at $x \in \mathcal{X}$

\emph{Dynamic systems:} We aim to control an agent modeled as a control-affine dynamic system:
\begin{equation}
	\dot{x}=f_x(x,u)=f(x) + G(x)u, 
	\label{eqn:agent_dyn}
\end{equation}
where $u \in \real^m, \ G : \mathcal{X} \rightarrow \real^{n \times
	m},\ f : \mathcal{X} \rightarrow \real^n$ and $f$ and $G$ are
locally Lipschitz.  We also consider a dynamic risk
% $y \in \mathcal{X}$:
\begin{equation}
	\dot{y}=f_y(y), \ y \in \mathcal{X}, \ f_y : \mathcal{X} \rightarrow \real^{n},
	\label{eqn:obs_dyn}
\end{equation}  
with a locally Lipschitz $f_y$. We focus on moving obstacles as the
source of risk, but the approach can be extended to other scenarios.
We also assume that a asymptotically stable controller $k:\mathcal{X}
\rightarrow \real^m$
%\margin{say where $k(x)$ lives} 
has been designed to guide the agent to a goal state $x^* \in
\mathcal{X}$
%\margin{what is $\matbb{X}$? if
%  you meant $\mathcal{X}$, please pay more attention to notational
%  consistency --- this is a first year student type of error (...)}
in the absence of risk sources.
We wish to drive the agent to a goal $x^*\in \mathcal{X}$ safely,
while avoiding risky areas. Formally, 
we define safety
considering a perceived spatial risk function $R_c$ as follows:

\begin{definition}(Perceived Safety) An agent moving
	under~\eqref{eqn:agent_dyn}, %in a state space $\mathcal{X}$
	%  \margin{and also refer here to risk dynamics} 
	and subject to an uncertain cost source $c$ with
	dynamics~\eqref{eqn:obs_dyn}, is said to be safe w.r.t.~the
	perceived risk $R_c$
	%  \margin{add ``wrt $R$'' }
	iff $R_c(\xi(t)) \leq \rho$,  $\forall \, t \ge 0$, for some
	tolerance $\rho \in \realpositive$.  
	%  \margin{ I think
	%    the cost itself is not dynamic as a function of $\xi = y - x$, but
	%    what is dynamic are $y$ and $x$. I'm just removing the word
	%    ``dynamic'' from $c$, no big deal.}
	\label{def:safety}	
\end{definition}
We now state the problems we address in this work:
%In this work, we will address the following two problems:
\begin{problem}(RPA safe sets) Given a risky environment
	$\mathcal{X}$, endowed with % \margin{I've changed
	%    ``embedded'' to ``endowed'', sounds better} 
	an uncertain cost $c$, design perceived safety sets considering RPMs
	from \eqref{eqn:risk_def}. Characterize and constrast
	the properties of these sets among the three RPMs.
	\label{prob:CPT_safety_set}
\end{problem} 
\vspace*{-0.3cm}
\begin{problem}(RPA safe controls) Under previous conditions, design a
	controller $u$, nominally deviating from a stable state feedback
	controller $k$, such that the agent reaches the goal $x^*$ safely
	(Definition~\ref{def:safety}) and examine feasibility of $u$.
	\label{prob:CPT_CBF}
\end{problem}

%%%%%%%%%%%%%%%%%%%%%%%%%%%%%%%%%%%%%%%%%%%%%%%%%%%%%%%%%%%%%%%%%%%%%%%%%%%%%%

\section{Perceived Safety using various RPMs}
\label{sec:risk_perception}
This section compares various RPMs, solving
Problem~\ref{prob:CPT_safety_set}.
Given an uncertain field cost $c$, we apply the different risk
perception models (see Section~\ref{sec:rpm}) to obtain the
corresponding fields, $R_c$.
%\margin{should it be on $\mathcal{Z}$?}
%representing the perceived spatial risk. 
With this, let us define the
following sets:
\vspace*{-0.5ex}
\begin{subequations}\label{eqn:safety_set}
	\begin{align}
		\subscr{\mathcal{X}}{safe} (R_c;y) =& \{x \in \mathcal{X} | R_c(y-x) \leq \rho \}, \label{eqn:safety_set1} \\
		%  \subscr{\mathcal{X}}{tol} =& \{x \in \mathcal{X} : R(\xi) = \rho \},\label{eqn:safety_set2}  \\
		\subscr{\mathcal{X}}{risky} (R_c; y) =& \{x \in \mathcal{X} | R_c(y-x) >
		\rho \} \label{eqn:safety_set3}.
	\end{align}
\end{subequations}
In particular, these sets depend on the choice of $R_c$ from
\eqref{eqn:risk_def}.  Given $\mathcal{M}$, we define the range set
$\mathcal{R}_\mathcal{M}(c )\subset \realnonnegative$ associated with 
$\mathcal{M}$ wrt $c$ as the set
$\mathcal{R}_{\mathcal{M}}(c) \triangleq\{r \in \real|r= R_\Theta(c),
\forall R_{\Theta} \in \mathcal{M} \}$\footnote{When clear from the
	context, we will just denote $\mathcal{R}_\mathcal{M}(c) \equiv
	\mathcal{R}$.}.
Fix a model $\mathcal{M}$ and a risk source at
$y \in \mathcal{X}$. The total safe set of $\mathcal{M}$ wrt $y$ is
given as $\mathcal{Y}_{\hspace*{-0.6ex} \mathcal{M}}(y,c) \triangleq \bigcup_{R_c \in
	\mathcal{M}} \subscr{\mathcal{X}}{safe}(R_c;y) $ (resp.~the total risky
set of $\mathcal{M}$ wrt $y\in \mathcal{X}$ is
$\overline{\mathcal{Y}}_{\hspace*{-1ex}\mathcal{M}}(y,c) \triangleq \bigcup_{R_c \in
	\mathcal{M}} \subscr{\mathcal{X}}{risky}(R_c;y)$). 
%\margin{I'm
%  including the dependency on $y$, and on $\mathcal{M}$} 
Thus, given $y \in \mathcal{X}$, the set
$\mathcal{Y}_{\hspace*{-0.6ex}\mathcal{M}}(y,c)$
(resp.~$\overline{\mathcal{Y}}_{\hspace*{-1ex}\mathcal{M}}(y,c))$ covers all
the states in $\mathcal{X}$ that safe (resp.~unsafe) according to a RPM
$\mathcal{M}$.
% \margin{we could also use a general $R$ here and then distinguish
% the types of sets as e.g. $\subscr{\mathcal{X}}{safe}(R)$? }
%Next, we compare the properties RPMs based on their
%corresponding safety sets.

\begin{definition}\longthmtitle{Inclusiveness and Strict Inclusiveness}
	\label{def:inclusiveness}
	Consider two RPMs $\mathcal{M}_1$ and $\mathcal{M}_2$, a
	threshold $\rho \in \realpositive$, and a risk source at $y
	\in \mathcal{X}$.  Let the sets 
	%        \margin{add the c in the Y
	%          from here}
	$\mathcal{Y}_1(y,c),\,\overline{\mathcal{Y}}_1(y,c)$ and
	$\mathcal{Y}_2(y,c),\, \overline{\mathcal{Y}}_2(y,c)$ be the
	total safe and risky sets of $\mathcal{M}_1$ and
	$\mathcal{M}_2$ wrt $y$ and a spatial cost $c$, respectively.
	We say that $\mathcal{M}_1$ is more inclusive than
	$\mathcal{M}_2$ ($\mathcal{M}_1 \vartriangleright
	\mathcal{M}_2$) if either $\overline{\mathcal{Y}}_2
	(y,c)\subseteq \overline{\mathcal{Y}}_1(y,c)$ and
	$\mathcal{Y}_2(y,c) \subsetneq \mathcal{Y}_1(y,c)$ holds, or
	$\overline{\mathcal{Y}}_2 (y,c)\subsetneq
	\overline{\mathcal{Y}}_1(y,c)$ and $\mathcal{Y}_2(y,c)
	\subseteq \mathcal{Y}_1(y, c)$ holds, for all $y \in
	\mathcal{X}$ and costs $c:\mathcal{Z} \rightarrow
	\realnonnegative$.
	%      \margin{the
	%          definitions right now depend on the choice of $c$. To make
	%          this clear, we can also show how the safe sets depend on $c$
	%          in the notation (as in the blue part)}
	%        \margina{So should we propogate the dependence of $c$ all the
	%          way to $R_c$ then?}  \margin{yes, let's do that} 
	If $\overline{\mathcal{Y}}_2(y,c) \subsetneq
	\overline{\mathcal{Y}}_1(y,c)$ and $\mathcal{Y}_2
	(y,c)\subsetneq \mathcal{Y}_1(y,c)$ both hold, then
	$\mathcal{M}_1$ is strictly more inclusive than
	$\mathcal{M}_2$ ($\mathcal{M}_1 \blacktriangleright
	\mathcal{M}_2$).
\end{definition} In particular, if $\mathcal{M}_1 \triangleright
\mathcal{M}_2$, then $\mathcal{M}_1$ results into a wider range
of safety and risky sets for a given environment than
$\mathcal{M}_2$.

Now we compare the inclusiveness of CPT, CVaR and ER via their
respective value functions. We start by comparing the range
space of these RPMs.

\begin{lemma}
	\label{lem:rangespace}
	% Consider a discrete random field cost $c$ over $\mathcal{Z}$.
	Consider a threshold $\rho \in \realnonnegative$, a risk source at
	$\bar{y} \in \mathcal{X}$, and two RPMs $\mathcal{M}_1,\,
	\mathcal{M}_2$ with range spaces $\mathcal{R}_1,\,
	\mathcal{R}_2$, respectively.
	If $\mathcal{R}_2 (c)\subseteq \mathcal{R}_1(c)$, and
	if there exists an $R_{1,c} \in \mathcal{M}_1$ such that $R_{1,c} >
	R_{2,c}$ or $R_{1,c} < R_{2,c}$ for any $R_{2,c} \in \mathcal{M}_2$, and
	any $c$, then $\mathcal{M}_1 \vartriangleright
	\mathcal{M}_2$. In addition, if there are $R_{1,c}^a, R_{1,c}^b \in
	\mathcal{M}_1$ such that $R_{1,c}^a > R_{2,c}^a $ and $R_{1,c}^b <
	R_{2,c}^b $,  $\forall \, R_{2,c}^a,R_{2,c}^b \in \mathcal{M}_2$, and any  $c$,
	then $\mathcal{M}_1 \blacktriangleright \mathcal{M}_2$.
	%  \margin{I've changed the superindices by a, b,
	%    propagate the change}
\end{lemma}
\begin{proof}
	Fix  $c$. Since $\mathcal{R}_2 \subseteq \mathcal{R}_1$,
	$\forall \, R_{2} \in \mathcal{M}_2$, there is  $R_1 \in \mathcal{M}_1$
	s.t. $R_1(c(\bar{y}-x))=R_2(c(\bar{y}-x))$,  $\forall \,x \in
	\mathcal{X}$. 
	Thus,
	$\mathcal{Y}_2(\bar{y}, c)
	\subseteq \mathcal{Y}_1(\bar{y}, c)$ and
	$\overline{\mathcal{Y}}_2(\bar{y}, c)
	\subseteq \overline{\mathcal{Y}}_1(\bar{y}, c)$.  
	Assume  $\exists \overline{R}_1,
	\widetilde{R}_1 \in \mathcal{M}_1$ s.t.~$\overline{R}_1(c(\bar{y}-x))>R_2(c(\bar{y}-x))$ or
	$\widetilde{R}_1(c(\bar{y}-x))<R_2(c(\bar{y}-x))$ 
	%  \margin{for
	%    example, this property is stronger than the range set inclusion
	%    because the functions are for all $x$ the same} 
	hold for
	all $R_2 \in \mathcal{M}_2$. % and $x \in \mathcal{X}$. 
	%\margin{I
	%  guess I don't know what we need this latter condition if the first
	%  one already seems to imply the result. } % Since $c$ and
	% $\bar{y}$ are fixed, the above conditions correspondingly applies to
	% all $x \in \mathcal{X}$.
	This implies either $\overline{\mathcal{Y}}_2(\bar{y}, c) \subsetneq
	\overline{\mathcal{Y}}_1(\bar{y}, c)$ or $\mathcal{Y}_2
	(\bar{y}, c)\subsetneq \mathcal{Y}_1 (\bar{y}, c)$. Inclusiveness follows
	from Definition~\ref{def:inclusiveness}. In parallel, $\mathcal{M}_1 \blacktriangleright
	\mathcal{M}_2$.
\end{proof}
\begin{lemma}
	\label{lem:rpm_rangespace}
	% Let $c:\mathcal{S} \rightarrow \realnonnegative$ be a discrete cost
	% random variable. 
	Consider the CPT, CVaR and ER risk models, with associated range
	sets $\subscr{\mathcal{R}}{CPT}(c),\,
	\subscr{\mathcal{R}}{CVaR}(c)\,,$ and $
	\subscr{\mathcal{R}}{ER}(c)$. Then, it holds that
	$\subscr{\mathcal{R}}{CPT}(c) \supsetneq
	\subscr{\mathcal{R}}{CVaR}(c) \supseteq
	\subscr{\mathcal{R}}{ER}(c)$, $\forall\, c$.
\end{lemma}
\begin{proof}
	Fix $c$. Note that $\subscr{\mathcal{R}}{ER}(c) =\{c_\mu\}$.  By
	choosing $\supscr{R}{cpt}_{\overline{\theta}} \in \text{CPT} $ with
	$ \overline{\theta}= \{1,1,1,1\} $ and $\supscr{R}{CV}_0 \in
	\text{CVaR} $ we have
	$\supscr{R}{CV}_0(c)=\supscr{R}{ER}(c)=\supscr{R}{cpt}_{\overline{\theta}}(c)$,
	$\forall \,c$. Note that only if $c_\sigma=0$ then
	$\supscr{R}{CV}_q(c)=c_\mu=\supscr{R}{ER}(c) $ for all $q$. When
	$c_\sigma \neq 0$, with any other valid choice of parameters $q$ in
	CVaR %and $\theta$ in CPT,
	we obtain $\supscr{R}{CV}_q(c) \notin
	\subscr{\mathcal{R}}{ER}(c)$. We can find $\theta \neq
	\overline{\theta}$ such that $\supscr{R}{cpt}_\theta(c) \notin
	\subscr{\mathcal{R}}{ER}(c)$, $\forall \,c$.  Hence,
	$\subscr{\mathcal{R}}{ER}(c) \subseteq
	\subscr{\mathcal{R}}{CVaR}(c)$ and $\subscr{\mathcal{R}}{CPT}(c)
	\supsetneq \subscr{\mathcal{R}}{ER}(c)$.
	
	%  Fix parameter $q \in [0,1]$.
	%\margin{rewrite sentence, does not make sense}
	For CVaR, $\supscr{R}{CV}_0=\{c_\mu\}$ and $\supscr{R}{CV}_1=\{b\}$,
	where $b \in \real$ is the worst-case outcome of $c$. Since
	$\supscr{R}{CV}_q$ increases in $q$, 
	$\subscr{\mathcal{R}}{CVaR} \subseteq [c_\mu,b]$.  Choosing
	$\theta_1 = \{1,1,1,\lambda\},$ for $\lambda \geq 1$, leads to
	$\supscr{R}{cpt}_{\theta_1}(c)= \lambda \sum_i c_i p_i = \lambda c_\mu$. Taking
	$\lambda \in [1,\bar{b}]$, with $\bar{b} > \frac{b}{c_\mu}$, we get
	$\subscr{\mathcal{R}}{CPT}(c) \supset [c_\mu,b]$; hence,
	$\subscr{\mathcal{R}}{CPT}(c) \supsetneq
	\subscr{\mathcal{R}}{CVaR}(c)$.
	%	We note that the range spaces $\subscr{\mathcal{R}}{CPT}$,
	% $\subscr{\mathcal{R}}{CVaR}$ and $\subscr{\mathcal{R}}{ER}$ are
	% all intervals on the real line and thus are all simply connected
\end{proof}

The previous results now lead to the following.
\begin{theorem}
	\label{thm:inclusiveness}
	Let $c$ %$c:\mathcal{Z} \rightarrow \realnonnegative$
	be a discrete random field cost. Consider the ER, CVaR and CPT
	risk perception models with risk value functions
	$\supscr{R}{ER}$, $\supscr{R}{CV}_q$, and
	$\supscr{R}{cpt}_\theta$, respectively. For any threshold $\rho
	\in \realnonnegative$ and risk source $\bar{y} \in
	\mathcal{X}$, $\text{CPT} \vartriangleright \text{CVaR}$ and $
	\text{CPT} \vartriangleright \text{ER} $ holds.  If the cost
	outcomes are strictly lower-bounded by 1, %satisfy $c_i \geq 1, \forall i $,
	%        \margin{use the shorthand command
	%          $\until{M}$, the dots are formatted better in math mode}
	then $\text{CPT} \blacktriangleright \text{CVaR}$ and $
	\text{CPT} \blacktriangleright \text{ER} $. If in fact
	$c_\sigma(\bar{y}-x) > 0, \forall x \in \mathcal{X}$, then
	$\text{CPT} \blacktriangleright \text{CVaR} \vartriangleright
	\text{ER} $.
\end{theorem}

\begin{proof}
	From Lemma~\ref{lem:rpm_rangespace},
	$\subscr{\mathcal{R}}{ER}(c) \subsetneq \subscr{\mathcal{R}}{CPT}(c)$ and
	$\subscr{\mathcal{R}}{CVaR}(c) \subseteq \subscr{\mathcal{R}}{CPT}(c)$.
	As in Lemma~\ref{lem:rpm_rangespace}, % choose $\theta_1
	% = \{1,1,1,\lambda\}, \lambda \geq 1$, to define
	take $\supscr{R}{cpt}_{\theta_1} = 
	\lambda c_\mu$, for some $\theta_1$. Choosing $\lambda = \bar{b}$,
	with $\bar{b} > \frac{b}{c_\mu}$, we get
	$\supscr{R}{cpt}_{\theta_1} > \supscr{R}{CV}_q$, for any $q\in
	[0,1]$, and $ \supscr{R}{cpt}_{\theta_1} > \supscr{R}{ER}$.
	Thus, from Lemma~\ref{lem:rangespace}, we have $\text{CPT}
	\vartriangleright \text{CVaR}$ and $ \text{CPT} \vartriangleright
	\text{ER} $.
	Now assume $c_i >1$ for all $i \in \until{M}$. Taking $\theta_1 =
	\{1,1,1,\lambda\},$ with $ \lambda > 1$, we get
	$\supscr{R}{cpt}_{\theta_1} > \supscr{R}{CV}_q $ for any $q\in
	[0,1]$ and $ \supscr{R}{cpt}_{\theta_1} > \supscr{R}{ER} $. Now,
	take $\theta_2 = \{1,1,\gamma,1\},$ with $ 0 < \gamma < 1$, we have
	$\supscr{R}{cpt}_{\theta_2}(c)= \sum_i c_i^{\gamma}p_i $. Since
	$c_i,p_i >0$, $ \forall i$,
	%\margin{please break down equations (I've done it)
	%  otherwise there are no spaces between them} 
	then $\supscr{R}{cpt}_{\theta_2}(c) < \sum_i c_i p_i $, implying
	$\supscr{R}{cpt}_{\theta_2}(c) < c_\mu$ ando
	$\supscr{R}{cpt}_{\theta_2}(c) < \supscr{R}{CV}_{q}(c)$, $\forall \,
	q\in [0,1]$.  From Lemma~\ref{lem:rangespace}, we get $\text{CPT}
	\blacktriangleright \text{CVaR}$ and $ \text{CPT}
	\blacktriangleright \text{ER} $.
	
	Finally, assume $c_\sigma > 0$.  There is $q \in (0,1)$ such that
	$\supscr{R}{CV}_q(c)>\supscr{R}{ER}$. Since the lower bound of
	$\supscr{R}{CV}_q(c)$ is $c_\mu=\supscr{R}{ER}$, there is no $q$
	s.t.~$\supscr{R}{CV}_q(c)<\supscr{R}{ER}$. Hence from
	Lemma~\ref{lem:rangespace} and the first part of this result, we get
	$\text{CPT} \blacktriangleright \text{CVaR} \vartriangleright
	\text{ER} $.
\end{proof}
The above arguments show CPT can produce a larger variety of safe and
risky sets leading to richer risk perception. This is illustrated via
simulations in Section~\ref{sec:results}.

%{\color{orange}
\emph{Additional properties of RPMs: }
% Previously, we defined the inclusiveness properties to compare two
% different RPMs in the context of safety. Now, we define additional
% risk aversion and risk sensitivity properties that characterize the
% performance of an RPM in the context of safety.
In addition to the notion of inclusiveness, we now characterize
the \textit{versatility} of a RPM in the context of perceived safety.
\begin{definition}\longthmtitle{Versatility of a RPM}
	\label{def:rpm_versatility}
	Consider a compact space $\mathcal{X}$,  a risk source
	$\bar{y} \in \mathcal{X}$, and a discrete random
	field cost $c$, with range in $[ \subscr{c}{min},
	\subscr{c}{max}] \subseteq \realnonnegative$.  Let
	$\overline{I}$ be a 
	compact interval. An RPM $\mathcal{M}$ is said
	to be $\overline{I}-$versatile if $\{x \in \mathcal{X} \,|\,
	c(\bar{y} - x) \le \subscr{c}{$\ell$}\} \subseteq
	\mathcal{Y}_{\mathcal{M}}$ for any $\subscr{c}{$\ell$} \in
	\overline{I}$ for a given $\rho > 0$. If $\overline{I}
	\supseteq [\subscr{c}{min},\subscr{c}{max}]$, then
	$\mathcal{M}$ is  \emph{most versatile} in
	$\mathcal{X}$.
\end{definition}
The above definition implies that an RPM is $\overline{I}-$versatile,
if it has a risk-perception function%  for each value $c_\ell$ in
% $\overline{I}$, and it
that perceives any states having costs less than $c_\ell$ as safe,
$\forall \, c_\ell$. Further, $\mathcal{M}$ is \emph{most versatile}
when it contains risk-perception functions that capture a range of
perceptions from most risk averse (only states having costs $c \leq
\subscr{c}{min}$ are safe) to the least risk-sensitive (every state
including states having the highest cost $\subscr{c}{max}$ as safe).
With this, we will look at versatility of the three RPMs.

\begin{lemma}
	\label{lem:CPT_risk_averse}
	Consider a compact space $\mathcal{X}$, with a risk source
	$\bar{y} \in \mathcal{X}$, and associated discrete random
	field cost $c$.
	Then, CPT can capture the most risk averse perception, i.e.~the set $\{x
	\in \mathcal{X} | c(\bar{y} - x) \le \subscr{c}{min}\}$ is
	considered safe.
\end{lemma} 

\begin{proof}
	Choosing $\theta$ as in Lemma~\ref{lem:rpm_rangespace} and 
	$\lambda=\frac{\rho}{\subscr{c}{min}},$ the result follows from
	\eqref{eqn:safety_set1} and
	Definition~\ref{def:rpm_versatility}.
\end{proof}
\begin{proposition}
	\label{prop:versatility}
	Under the setting of Lemma~\ref{lem:CPT_risk_averse}, CPT can
	capture the least risk sensitive perception (the set $\{x \in
	\mathcal{X} | c(\bar{y} - x) \le \subscr{c}{max}\}$ is
	considered safe), if $\subscr{c}{min} \geq 1$, for $\rho \geq
	1$, and
	%         \margin{in the proof we need
	%          $\rho >1$ which one is it?}
	$\forall \, i $ over $\mathcal{X}$. Consequently, CPT is
	\emph{most versatile} in $\mathcal{X}$.
\end{proposition}
\begin{proof}
	For $\theta_2= \{1,1,\gamma,1\},$ with $ 0 \leq \gamma \leq 1$ we
	have $\supscr{R}{cpt}_{\theta_2}(c)= \sum_i c_i^{\gamma}p_i $. Now,
	choosing $\gamma < \frac{\log \rho}{\log \subscr{c}{max}}$, since
	$c_i \geq 1 ,p_i \geq 0$, $ \forall i$, and $\rho \geq 1$, we get
	$\supscr{R}{cpt}_{\theta_1} \leq \rho$. Thus, from
	\eqref{eqn:safety_set1} and Definition~\ref{def:rpm_versatility},
	the first result follows. Take now $\theta_1 = \{1,1,1,\lambda\}$
	and $\theta_2$. Observe that $\supscr{R}{cpt}_{\theta_1}$ is
	continuous in $\lambda$ and $\supscr{R}{cpt}_{\theta_2}$ is
	continuous in $\gamma$. By the  intermediate
	value theorem $\exists\, \lambda$ s.t.~$\supscr{R}{cpt}_{\theta_1}
	\in [c_\mu, \subscr{c}{max}]$, and a $\gamma$
	s.t.~$\supscr{R}{cpt}_{\theta_2} \in [\subscr{c}{min},
	c_\mu]$. Hence, from Lemma~\ref{lem:CPT_risk_averse}, %  and
	% Definition~\ref{def:rpm_versatility}
	CPT is \emph{most
		versatile} in $\mathcal{X}$.
\end{proof}

\begin{lemma}
	Under the assumptions of Lemma~\ref{lem:CPT_risk_averse}, with
	$\overline{I}_1 = [c_\mu,\subscr{c}{max}]$ and $\overline{I}_2 =
	\{c_\mu \}$, CVaR is $\overline{I}_1-$versatile and ER is
	$\overline{I}_2-$versatile. Hence neither are \emph{most
		versatile} RPMs.
\end{lemma} 
\begin{proof}
	This result trivially follows from the range spaces
	$\subscr{\mathcal{R}}{CVaR}(c)$ and $\subscr{\mathcal{R}}{ER}(c)$ in
	the proof of Lemma~\ref{lem:rpm_rangespace}.
\end{proof}

\section{Control design with Risk-Perception-Aware-CBFs}
\label{sec:control}
Here, we address Problem~\ref{prob:CPT_CBF} and design controls $u$
for an agent subject to~\eqref{eqn:agent_dyn}, to ensure perceived
safety (Definition~\ref{def:safety}). 
To do this, we formally adapt CBFs (see~\cite{ADA-SC-ME-GN-KS-PT:19}) to our setting.
%see~\cite{ADA-SC-ME-GN-KS-PT:19} for more properties of CBFs.
\begin{definition}[RPA-CBF] Consider an agent subject
	to~\eqref{eqn:agent_dyn}, a dynamic source of
	risk~\eqref{eqn:obs_dyn}, and a perceived risk $R_c$ model. A
	$\mathcal{C}^1$ function $h_R \triangleq h \circ R_c :
	\mathcal{Z} \rightarrow \real$
	%  \margin{I've
	%    changed this, I think this is what you really meant.}
	is an RPA-CBF for this system, if there is an extended class
	$\mathcal{K}_\infty$ function $\eta_1$ such that the control set
	$\subscr{K}{R}$ defined as
	\begin{equation}
		\subscr{K}{R}(R_c)= \{u \in \mathcal{U}|\dot{h}_R(\xi)
		\geq -  \eta_1(h_R(\xi)))\},
		\label{eqn:k_cbf_set}
	\end{equation}
	is non-empty for all $\xi \in \mathcal{Z}$.  % margin{Is the
	\label{def:cbf}
\end{definition}
The existence of $h_R$ according to Definition~\ref{def:cbf} implies
that the superlevel set $\{x \in \mathcal{X}| h_R(\xi) > 0 \}$ is
forward invariant under~\eqref{eqn:agent_dyn}.  We specify $h_R = h
\circ R_c$ via $h$ given as
\begin{equation}
	h(\xi) \triangleq \eta_2(\rho-R_c(\xi)), 
	\label{eqn:cbf_cpt}
\end{equation}
where $\eta_2 : \real \rightarrow \real$  is a $\mathcal{C}^1$ extended class
$\mathcal{K}_\infty$ function.
%\begin{remark}
Since $\eta_2$ is non-decreasing,  $h(R_c(\xi)) \geq 0$ implies 
$R_c(\xi) \leq \rho $ and from \eqref{eqn:safety_set1},
$\subscr{\mathcal{X}}{safe} (R_c;y)= \{x \in
\mathcal{X} | h(R_c(y-x)) \geq 0
\}$. Thus,  $h( R_c(\xi))>0$ indicates that $x$ is 
perceived as safe w.r.t.~$R_c$.

The RPA control input $u$ can be now computed via:
%\margin{here you use
%  $\eta$ and not $\eta_2$, pls decide} 
\vspace*{-0.2cm}
{\small 
	\begin{subequations} \label{eqn:clbf_opt}
		\begin{align} \label{eqn:clbf_opt_1}
			& u(x)=\  \argmin_{u} \  \|u-k(x) \|^2 \\
			\label{eqn:clbf_opt_2} & \text{s.t. }\ \frac{d\eta_2}{d
				R_c}\left( \frac{\partial R_c}{\partial
				\xi}(\xi)\right) \cdot \left(f_y(y)-f_x(x)\right)
			\geq - \eta_{1}( h_R(\xi)).
		\end{align}
\end{subequations}}
The above problem captures the notion of minimally modifying a stable
controller to ensure safety of the system. Next we will analyze the
feasibility conditions for the proposed controller $u$ and compare
it across the proposed  models.

\emph{Feasibility analysis and comparison:} We first describe a
construction of finite outcomes of $c$ from $c_\mu$ and $c_\sigma$
called ``truncated-Gaussian cost'' which will be used for
analysis. Assume that $c(\xi)$ is distributed as a truncated
Gaussian\footnote{This  truncation reassigns the probability
	mass s.t.~$c(\xi) \in [c_\mu (\xi) - 3c_\sigma(\xi),c_\mu (\xi) +
	3c_\sigma(\xi)] $ % is is based
	% on the $3\sigma$ rule of thumb to describe the support of the
	% distribution and probability mass reassigned
	using an appropriate
	re-normalization constant. }
$\mathcal{N}_T(c_\mu (\xi),c_\sigma(\xi)^2)$. Then, given $M \in
\mathbb{N}$, we approximate $c$ by means of $M$ discrete values $c_i$,
$i \in \until{M}$, with probability  calculated from the
CDF $F$ of $c$ at each $c_i$. That is, $p_1 = F(c_1)$, and $p_i =
F(c_i) - F(c_{i-1})$, for $i \in \{2,\dots,M\}$.
Now, we show conditions on $u$ for the set $\subscr{K}{CBF}$ to be
non-empty for a given risk function $R_c$.  We first define a few
constants and variables to help us compare the feasibility conditions
of the three RPMs.  Let $\phi_\xi \in [-\pi,\pi]$ be the relative
angle\footnote{recall angle between two vectors $a,b \in \real^n$ is
	given by $\phi = \cos^{-1} \big(\frac{a \cdot b}{\|a\| \|b\|}\big)$}
between $\frac{\partial R_c}{\partial \xi}$ and $\dot{\xi}(u;x,y)$,
and $c_\mu'=\frac{d c_\mu}{d \xi}$ and $c_\sigma'=\frac{d c_\sigma}{d
	\xi}$. Now define $k^e(\xi)=(\eta_1\circ
\supscr{R}{ER}(\xi))/\frac{d\eta_2}{d \supscr{R}{ER}}$,
$k^v_q(\xi)=(\eta_1\circ \supscr{R}{CV}_q(\xi))/\frac{d \eta_2}{d
	\supscr{R}{CV}_q}$ and $k^v_\sigma=\frac{\mathbb{P}(F^{-1}(q))}{q}$.
Also we define constants $k^c_\theta(\xi)=(\eta_1\circ
\supscr{R}{cpt}_\theta(\xi))/\frac{d \eta_2}{d
	\supscr{R}{cpt}_\theta}$, $k^c_\mu= \lambda \gamma \sum_{i=1}^{M}
c_i^{\gamma-1} \Pi_i$ and $k^c_\sigma = \lambda \gamma
\sum_{i=1}^{M}\left(3 - \frac{6i}{M} \right)\left(c_i
\right)^{\gamma-1} \Pi_i$. Consider $\eta^e= \frac{k^e(\xi)}{\left\|
	c_\mu' \right\|}$, $\eta^v=\frac{k^v_q(\xi)}{\left\| c_\mu' +
	k^v_\sigma c_\sigma' \right\|}$ and
$\eta^c=\frac{k^c_\theta(\xi)}{\left\| k^c_\mu c_\mu' + k^c_\sigma
	c_\sigma' \right\|}$. The following holds.

\begin{proposition}
	\label{prop:feasibility}
	Let an agent and risk source be subject
	to~\eqref{eqn:agent_dyn} and \eqref{eqn:obs_dyn},
	respectively. Consider cost $c$ build from a truncated
	Gaussian field. If there is a $u$ s.t.:
	\begin{equation}
		\label{eqn:angle_condition}
		\|\dot{\xi}(u;x,y)\| \cos(\phi_\xi) 
		\geq - 
		\left( \dfrac{\eta_1( h_R(\xi))}{\frac{d\eta_2}{d R}
			\left\|\frac{\partial R_c}{\partial \xi}(\xi)\right\|}
		\right),
	\end{equation}
	then $h_R$ defined according to~\eqref{eqn:cbf_cpt} is a valid
	RPA-CBF for any $\eta_2$ and \eqref{eqn:clbf_opt} is
	feasible. Specifically, with $\tilde{\xi}=\|\dot{\xi}(u;x,y)\|
	\cos(\phi_\xi)$, the RHS of the above inequality reduces to $-
	\eta^e$, $- \eta^v$, and $- \eta^c$ for ER, CVaR and CPT,
	respectively.
\end{proposition}
\begin{proof}
  For first part, rearranging terms in~\eqref{eqn:angle_condition} we get:
\begin{equation}
  \kappa\left\| \frac{\partial R}{\partial \xi}(\xi)\right\| 
  \cdot \left\|f_y(y)-f_x(x,u)\right\|\cos(\phi_\xi) \geq - \eta_{1}(h(\xi)),
	\label{eqn:angle_condition2}
\end{equation} 
where $\kappa=\frac{d\eta_1}{d R}$. For the RPA-CBF to be valid, the set
$\subscr{K}{CBF}$ needs to be
non-empty.  Due to the dynamics of the agent and obstacle, $c_\mu$ and
$c_\sigma$ have dynamics:
\begin{equation}\label{eqn:mean_sd_dynamics}
  \dot{c}_\mu=\frac{\partial c_\mu}{\partial \xi}(f_y(y)-f_x(x,u)),\ 
  \dot{c}_\sigma=\frac{\partial c_\sigma}{\partial \xi}(f_y(y)-f_x(x,u)) .	
\end{equation}

Using the chain rule, we get the time derivative of $ h_R(\xi)$: 
\begin{subequations}\label{eqn:h_dot}
	\begin{align}\label{eqn:h_dot1}
		 \dot{h}_R(x,y,\xi,u)=& \frac{d\eta}{d R_s}\dot{R_s}(\xi), \\ \label{eqn:h_dot2}
		=& \frac{d\eta}{d R_s}\begin{bmatrix}\frac{\partial R_s}{\partial c_\mu}(\xi) \\ \frac{\partial R_s}{\partial c_\sigma}(\xi) \end{bmatrix}^\top
		\begin{bmatrix}\dot{c}_\mu(x,y,\xi,u) \\ \dot{c}_\sigma(x,y,\xi,u) \end{bmatrix} ,\\ \label{eqn:h_dot3}
		=&\kappa\left( \frac{\partial R_s}{\partial \xi}(\xi)\right) \cdot \left(f_y(y)-f_x(x,u)\right), \\ \label{eqn:h_dot4} 
		=&\kappa\left\| \frac{\partial R_s}{\partial \xi}(\xi)\right\| \cdot \left\|f_y(y)-f_x(x,u)\right\|\cos(\phi_\xi)
	\end{align}
\end{subequations}  

For the last part, the expressions are obtained by substituting the respective risk
functions and evaluating the partial derivatives $\frac{\partial
	R}{\partial c_\mu}$ and $\frac{\partial R}{\partial c_\sigma}$
(part of $\frac{\partial R}{\partial \xi}$). 
Thus we need to show the following hold true:
\begin{subequations}
	\begin{align}\label{eqn:angle_condition_er}
		\tilde{\xi}\geq& - 
		\left(\frac{k^e(\xi)}{\left\| c_\mu' \right\|}\right), 
		\text{  for ER,}\\ 
		\label{eqn:angle_condition_cvar}
		\tilde{\xi} \geq& - 
		\left(\frac{k^v_q(\xi)}{\left\| c_\mu'  
			+ k^v_\sigma c_\sigma' \right\|}\right), \text{  for CVaR,}\\
		\label{eqn:angle_condition_cpt}
		\tilde{\xi} \geq& - \left(\frac{k^c_\theta(\xi)}{\left\| k^c_\mu c_\mu' + k^c_\sigma c_\mu' \right\|}\right), \text{  for CPT}.		 
	\end{align}  
\end{subequations}
  For ER we get
  $\frac{\partial R}{\partial c_\mu}=1$ and $\frac{\partial
    R}{\partial c_\sigma}=0$.  For CVaR, since $c$ is assumed to
  belong to a truncated Gaussian distribution, we can use the closed
  form expression of CVaR \eqref{eqn:cvar_closed_form} for a Gaussian
  distribution to calculate the partials $\frac{\partial R^v}{\partial
    c_\mu}$ and $\frac{\partial R^v}{\partial c_\sigma}$ .
  \begin{equation}
    R^v_q= c_\mu + c_\sigma \big( \frac{\mathbb{P}(F^{-1}(q))}{q} \big).
    \label{eqn:cvar_closed_form}
  \end{equation}
  From \eqref{eqn:cvar_closed_form}, it is easy to see that CVaR is
  linear in $c_\mu$ and $c_\sigma$.  With this, we get $\frac{\partial
    R}{\partial c_\mu}=1$ and $\frac{\partial R}{\partial
    c_\sigma}=\frac{\mathbb{P}(F^{-1}(q))}{q}$.
  
  Substituting these derivatives in \eqref{eqn:h_dot} correspondingly
  for ER and CVaR, and using \eqref{eqn:mean_sd_dynamics} we obtain
  the results.
  
  For CPT, the expression is obtained by substituting the CPT risk function
  $\supscr{R}{cpt}_\theta$ and evaluating the partial derivatives $\frac{\partial
  	\supscr{R}{cpt}_\theta}{\partial c_\mu}$ and $\frac{\partial
  	\supscr{R}{cpt}_\theta}{\partial c_\sigma}$.  Constructing truncated Gaussian
  costs $c$ from $c_\mu$ and $c_\sigma$, we get outcomes $\{c_1,\dots,
  c_M \}$ and corresponding probabilities $\{p_1,\dots, p_M \}$
  resulting in constant $\Pi$ throughout.  In this way, from
  \eqref{eqn:cpt_risk}, the CPT value of a random cost $c$ with mean
  $c_\mu$ and $c_\sigma$ is given by:
  \begin{equation}
  	\label{eqn:app_CPT_value}
  	\supscr{R}{cpt}(c_\mu,c_\sigma)=\sum_{i=1}^{M} \lambda \left(c_\mu + c_\sigma\left(3 - \frac{6i}{M} \right)  \right)^\gamma \Pi_i.
  \end{equation}
  With this expression, we can proceed to calculate the partial
  derivatives $\frac{\partial \supscr{R}{cpt}_\theta}{\partial \mu}$ and
  $\frac{\partial \supscr{R}{cpt}_\theta}{\partial \sigma}$.
  %	\emph{Evaluating $\frac{\partial \supscr{R}{cpt}}{\partial\mu}$:} 
  From~\eqref{eqn:app_CPT_value}, we get
  \begin{subequations} \label{eqn:cpt_partials}
  	\small
  	\begin{align}
  		\label{eqn:app_partial_mu}
  		\frac{\partial \supscr{R}{cpt}}{\partial
  			c_\mu}(c_\mu,c_\sigma)=& \lambda \gamma
  		\sum_{i=1}^{M}\left(c_\mu + c_\sigma\left(3 -
  		\frac{6i}{M} \right) \right)^{\gamma-1}
  		\Pi_i, \\ \label{eqn:app_partial_sigma}
  		\frac{\partial \supscr{R}{cpt}}{\partial
  			c_\sigma}(c_\mu,c_\sigma)=& \lambda \gamma
  		\sum_{i=1}^{M}\left(3 - \frac{6i}{M}
  		\right)\left(c_\mu + c_\sigma\left(3 -
  		\frac{6i}{M} \right) \right)^{\gamma-1}
  		\Pi_i .
  	\end{align}
  \end{subequations}
  We have $\frac{\partial \supscr{R}{cpt}}{\partial
  	c_\mu}(c_\mu,c_\sigma)=k^c_\mu$ and $\frac{\partial
  	\supscr{R}{cpt}}{\partial
  	c_\sigma}(c_\mu,c_\sigma)=k^c_\sigma$. Substituting
  $k^c_\theta$, $k^c_\mu$ and $k^c_\sigma$ in
  \eqref{eqn:angle_condition}, we obtain
  \eqref{eqn:angle_condition_cpt}.
\end{proof}

From~\eqref{eqn:angle_condition}, observe that the RHS is independent
of $u$ and the LHS is independent of $R_c$ and the RPM. This
separation makes it easier to compare  various RPMs and
their associated feasibility conditions.
 
Next, we remark on the uncertainty perception of each RPM, which will
be used in the subsequent proposition to compare the size of control
sets $\supscr{K}{CBF}$ respectively generated by each of the RPMs.

\begin{remark}[Uncertainty perception among RPMs]
	\label{rem:uncertainty_perception}
	The ER model is insensitive to uncertainty as $\frac{\partial
		\supscr{R}{ER}}{\partial \sigma}=0$. In this way, CVaR is
	averse to uncertainty as $\frac{\partial
		\supscr{R}{CV}_q}{\partial \sigma} \geq 0$ for all $q$. With
	CPT, $\theta$ can be tuned to get both uncertainty insensitive
	and uncertainty averse behavior, additionally, it can also
	produce uncertainty liking behavior (when $\frac{\partial
		\supscr{R}{cpt}_\theta}{\partial \sigma} \leq
	0$). \footnote{The first two properties follow by choosing
		$\theta$ as in Theorem~\ref{thm:inclusiveness}. The latter
		property can be obtained by tuning the uncertainty
		perception parameters $\alpha$ and $\beta$. Since the chosen
		distribution is symmetric, we can examine the relation
		between $\Pi_i$ and $\Pi_{M-i}$ for $i \in
		\left(0,\frac{M}{2}\right)$. If we have $\Pi_i < \Pi_{M-i}$
		(for example when $w$ is concave) or $\Pi_i > \Pi_{M-i}$
		(when $w$ is convex), then we have $\frac{\partial
			\supscr{R}{cpt}}{\partial c_\sigma} > 0$, or
		$\frac{\partial \supscr{R}{cpt}}{\partial c_\sigma} < 0$,
		respectively.  A concave $w$ ($\alpha = 1, \beta <1$)
		implies that unlikely outcomes are viewed to be more
		probable compared with the more certain outcomes. This
		results into an ``uncertainty averse behavior'', which is
		reflected in the positive sign of $\frac{\partial
			\supscr{R}{cpt}}{\partial c_\sigma}$. Conversely, a convex
		$w$ ($\alpha = 1, \beta >1$) leads to an ``uncertainty
		liking behavior'' with $\frac{\partial
			\supscr{R}{cpt}}{\partial c_\sigma} < 0$.}.
\end{remark}
We finally compare the
the flexibility provided by each model via the
corresponding control sets $K$.
\begin{proposition}
\label{prop:safety_sets}
Assume the conditions of Proposition~\ref{prop:feasibility}
hold.  Then, the feasibility sets defined according to
\eqref{eqn:k_cbf_set} for the three RPMs satisfy $\subscr{K}{ER}
\subseteq \subscr{K}{CPT}$ and $\subscr{K}{CVaR} \subseteq
\subscr{K}{CPT}$.
\end{proposition}  
\begin{proof} In order to compare the feasibility of the sets
  $\subscr{K}{CBF}$ from~\eqref{eqn:k_cbf_set} for the three RPMs, we
  can compare their respective feasibility
  conditions~\eqref{eqn:angle_condition}. Consider $\eta^e=
  \frac{k^e(\xi)}{\left\| \frac{dc_\mu}{d \xi} \right\|}$,
  $\eta^v=\frac{k^v_q(\xi)}{\left\| \frac{dc_\mu}{d \xi} + k^v_\sigma
      \frac{dc_\sigma}{d \xi} \right\|}$ and
  $\eta^c=\frac{k^c_\theta(\xi)}{\left\| k^c_\mu \frac{dc_\mu}{d \xi}
      + k^c_\sigma \frac{dc_\sigma}{d \xi} \right\|}$. Since the LHS
  in~\eqref{eqn:angle_condition} remains the same for any RPM and its
  parameter choice, to prove the proposition, it is sufficient to show
  that $\eta^e \leq \eta^c$ and $\eta^v \leq \eta^c$. These
  inequalities follow from the choice of $\theta=\theta_1$ in
  Theorem~\ref{thm:inclusiveness} and CPT's more adaptable uncertainty
  perception from Remark~1.
 \end{proof} 
 It is interesting to note that although CVaR is more inclusive than
 ER as proved in Theorem~\ref{thm:inclusiveness}, it does not
 immediately translate into CVaR having a larger control feasibility
 set. We provide more insight in the following remark.
\begin{remark}
  Consider the control feasibility sets $\subscr{K}{ER}$ and
  $\subscr{K}{CVaR}$ respectively for ER and CVaR, defined according
  to~\eqref{eqn:k_cbf_set}. Then, depending on the choice of $q$ and
  construction of $c_\sigma$ we can obtain either $\subscr{K}{ER}
  \subseteq \subscr{K}{CVaR}$ or $\subscr{K}{CVaR} \subseteq
  \subscr{K}{ER}$. Looking at the LHS of inequalities
  \eqref{eqn:angle_condition_cvar} and \eqref{eqn:angle_condition_er},
  although we have $k_q^v (\xi) > k^e (\xi)$ from
  Theorem~\ref{thm:inclusiveness}, there isn't conclusive proof to
  suggest $\subscr{K}{ER} \subseteq \subscr{K}{CVaR}$ due to the
  additional $k^v_\sigma \frac{d c_\sigma}{d \xi}(\xi)$ term in the
  denominator of \eqref{eqn:angle_condition_cvar}.
\end{remark}
%\margin{read up to here}
\paragraph*{Stability analysis}
Next, let us look at the stability properties of the proposed
controller $u$ in~\eqref{eqn:clbf_opt}. It is clear that if the
nominal controller $k(x)$ also satisfies the safety
constraint~\eqref{eqn:clbf_opt_2}, then $u=k(x)$ and the stability
properties of $k(x)$ transfer over to $u$. To analyze stability, first
we look into the RPMs and determine how they affect the deviation from
$k(x)$. Later, we treat the controller $u$ as a perturbed version of
$k(x)$ and analyze accordingly.

Let $\delta=k(x)-u(x)$ be the perturbation to the nominal controller
$k(x)$ and $\supscr{\delta}{ER}$, $\supscr{\delta}{CV}_q$ and
$\supscr{\delta}{cpt}_\theta$ be the respective perturbations of ER,
CVaR and CPT with corresponding parameter choices. Then we have the
following:

\begin{proposition}
	\label{prop:stability_margin}
	Under the assumptions of Proposition~\ref{prop:feasibility},
	choose $u$ as in~\eqref{eqn:clbf_opt}. Assume
	$\|\supscr{\delta}{ER}\|$, $\|\supscr{\delta}{CV}_q\|$ and
	$\|\supscr{\delta}{cpt}_\theta\|$ are bounded. Then for any
	given states $x$, $y$, and choice of $q$, there exists a
	$\theta$ such that:
	\begin{enumerate}
		\item $\| \supscr{\delta}{cpt}_\theta\| \leq \| \supscr{\delta}{ER} \|$ and $\|
		\supscr{\delta}{cpt}_\theta\| \leq \| \supscr{\delta}{CV}_q \|$.
		\item The agent stabilizes inside
		$\mathbf{B}^{\epsilon^\mathcal{M}}(x^*)$ asymptotically for
		all RPMs and their respective $\epsilon$ follow
		$\supscr{\epsilon}{CPT}\leq \supscr{\epsilon}{CVaR}$ and
		$\supscr{\epsilon}{CPT} \leq \supscr{\epsilon}{ER}$.
	\end{enumerate} 
	
\end{proposition}
%\margina{summarized 2 lemmas, definition and 3 equations into 1 prop}
\begin{proof}
	For $1)$, apply~Proposition~\ref{prop:safety_sets} and the fact that
	$\subscr{K}{ER} \subseteq \subscr{K}{CPT}$ and $\subscr{K}{CVaR}
	\subseteq \subscr{K}{CPT}$.
	
	For $2)$, employ an ISS argument to construct the
	$\mathbf{B}^{\epsilon^\mathcal{M}}(x^*)$ for each RPM considering
	the unforced system with $u=k(x)$ in \eqref{eqn:agent_dyn} and
	$P(x)=G(x)\delta$ being the forcing term after applying RPA controls
	$u$ from \eqref{eqn:clbf_opt}. From ISS, since the radius of the
	stability ball is proportional to the upper bound on $\|P(x)\|$, the
	result immediately follows from the first part. 
\end{proof}
Proposition~\ref{prop:stability_margin} implies that, with an
appropriate  $\theta$, CPT can not only produce the least
perturbation among the three RPMs, but can also stabilize to the
smallest ball around~$x^*$.
\section{Simulation Results}
\label{sec:results}
Here, we visualize the results from Theorem~\ref{thm:inclusiveness}
and demonstrate the effectiveness of the controller generated in
\eqref{eqn:clbf_opt}. We consider a few scenarios involving an agent
moving in an 2D environment containing one or more moving obstacles
(sources of uncertain risk) and use this to compute the
RPA-CBF~\eqref{eqn:cbf_cpt} to guide the agent to a desired goal
safely. We compare CPT, CVaR and ER as RPA models and illustrate the
results followed by a discussion.

\emph{Uncertain Cost:}  We assume an agent $x \in
  \real^2$ with dynamics~\eqref{eqn:agent_dyn} in a 2D state space
  containing an obstacle $y \in \real^2$ moving according
  to~\eqref{eqn:obs_dyn}. We assume that the obstacle is imperfectly
  localized and is known to be within a ball of radius
  $\overline{r}$ centered at $y_\mu \in \mathcal{X}$, i.e , $y \in
  \mathbf{B}^{\overline{r}}(y_\mu)$\footnote{W.l.o.g.~this assumption
    also allows us to consider obstacles with a size.}.
 With this, the
  relative vector $\xi = y-x$ belongs to the space: $\xi \in
  \mathbf{B}^{\overline{r}}(y_\mu-x)$.  We use the notion of
  ``distance to endangerment (DTE)'', $d : \real^n \rightarrow
  \realnonnegative$, $d(\xi) \triangleq \| \xi \|$ to construct the
  uncertain cost $c$. From this, we obtain $d
  \in[\|x-y_\mu\|-\overline{r},\|x-y_\mu\|+\overline{r}]$.
  (visualized in Figure~\ref{fig:exp_setting1}). We consider the cost
$c(\xi)=~k_1 \mathrm{e}^{-k_2d(\xi)^2}$,
denoting the cost of being
at $x$, knowing 
the obstacle $y \in \mathbf{B}^{\overline{r}}(y_\mu)$, with constants $k_1,k_2 > 0$.
 \begin{figure}
 	\centering
 	\includegraphics[width=0.6\linewidth]{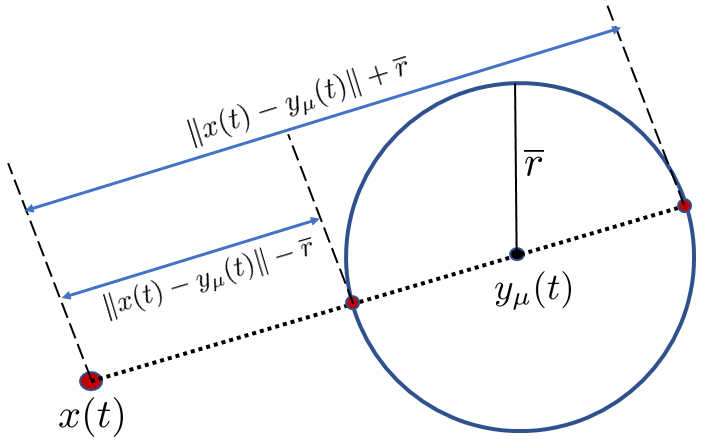}
 	\caption{Illustration of simulation setting and measuring DTE
          for an agent at $x(t)$, facing an obstacle which is
          localized imperfectly in a circle of radius $\overline{r}$
          and centered at $y_\mu(t)$. }
 	\label{fig:exp_setting1}
 \end{figure}

 With this, we assume the cost $c$ is distributed as a truncated
 Gaussian (Section~\ref{sec:risk_perception}) with $c_\mu(\xi)=k_1
 \mathrm{e}^{-k_2 d_\mu^2(\xi)}$
% \margin{$c_\mu$ is a function of $d$
%   but $d$ does not appear on the right hand side, only $\overline{d}$
%   does. Can you fix this? }
  and $c_{\sigma}(\xi)= c_\mu(\overline{r})
 p^\mathcal{N}(\xi,\mathbf{I})$, where $d_\mu=\|x-y_\mu\|$ and
 $p^\mathcal{N}(\mu,\Sigma)$ is the pdf of a bi-variate Normal distribution with mean
 $\mu$ and covariance $\Sigma$ and $\mathbf{I}$ is the 2D identity
 matrix.
 We proceed to construct the uncertain cost outcomes
 according to Section~\ref{sec:risk_perception} 
% \margin{we won't have a
%   label, so please fix} 
and then calculate $\supscr{R}{cpt},R^v,R^e$
 appropriately. We use the reference value
 $\rho=c_{\mu}(\overline{r})$, to denote the risk threshold.
%{\color{orange}}

 \begin{figure*}
 	\begin{subfigure}[t]{0.23\linewidth}
 		\includegraphics[width=\textwidth]{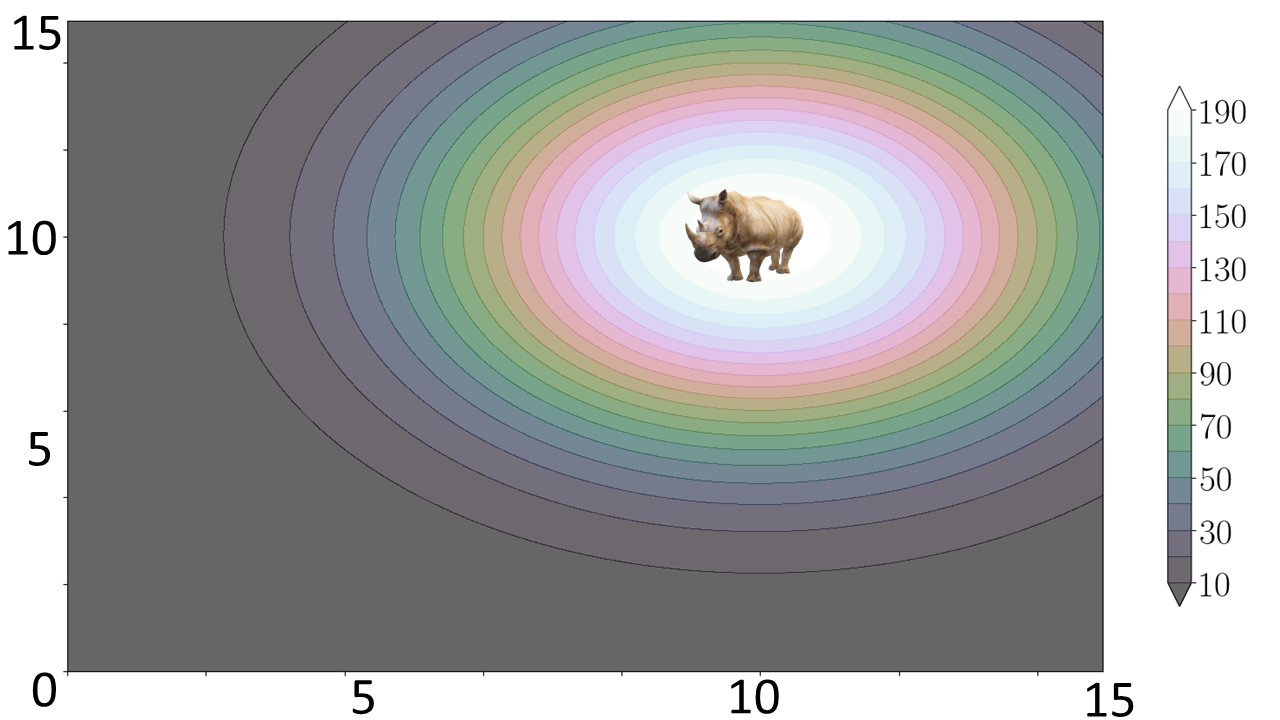}
 		\caption{Mean cost $c_\mu$ over $\mathcal{X}$}
 		\label{fig:c_mu}
 	\end{subfigure}%
 	~
 	\begin{subfigure}[t]{0.23\linewidth}
 		\includegraphics[width=\textwidth]{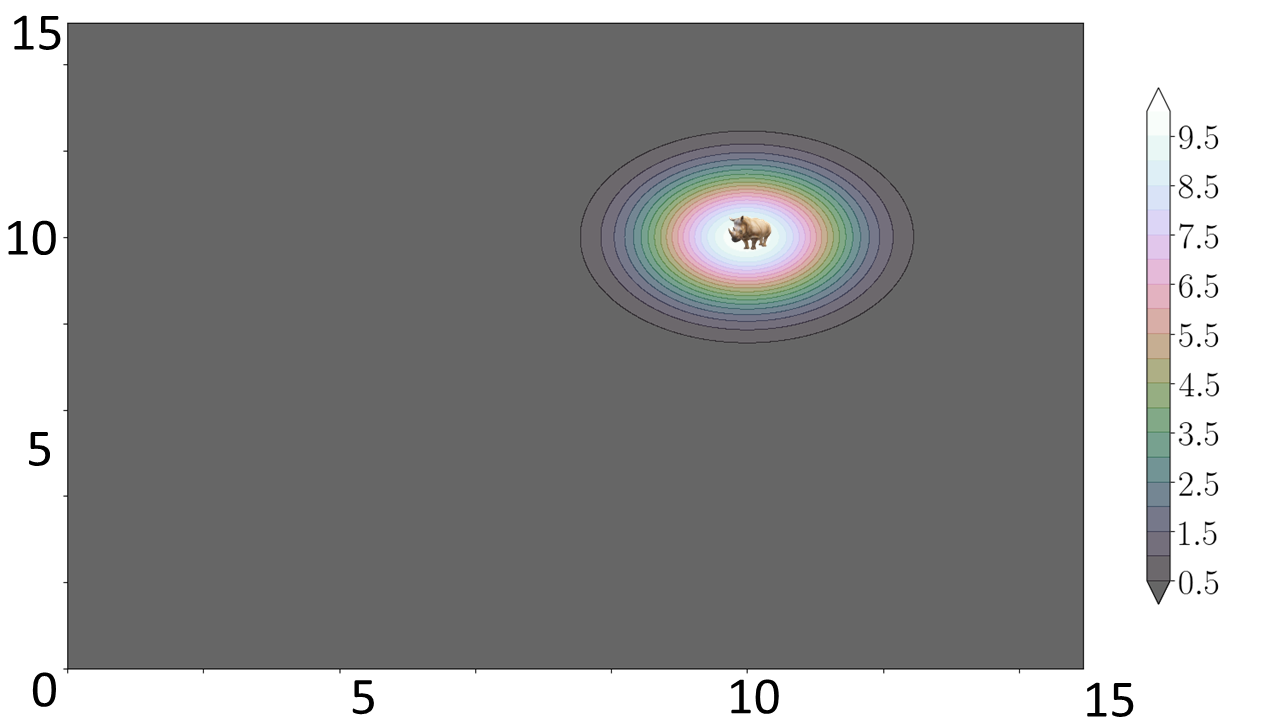}
 		\centering
 		\caption{Standard deviation $c_\sigma$ over $\mathcal{X}$}
 		\label{fig:c_sigma}
 	\end{subfigure}
 	~
 	\begin{subfigure}[t]{0.25\linewidth}
 		\includegraphics[width=\textwidth]{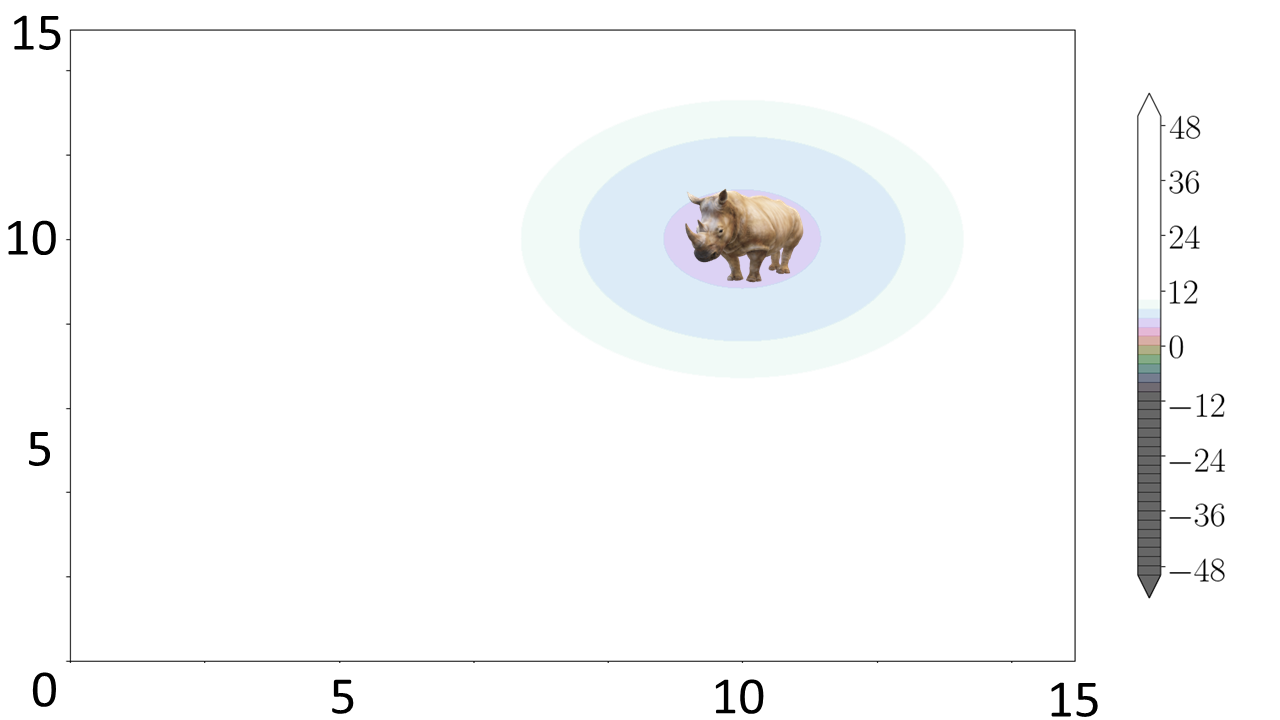}
 		\centering
 		\caption{Highly risk insensitive perception ($\gamma=0.45$ and $\rho=27$) with contour map depicting $\rho-\supscr{R}{cpt}$}
 		\label{fig:versatility_insensitive}
 	\end{subfigure}
 	~
 	\begin{subfigure}[t]{0.25\linewidth}
 		\includegraphics[width=\textwidth]{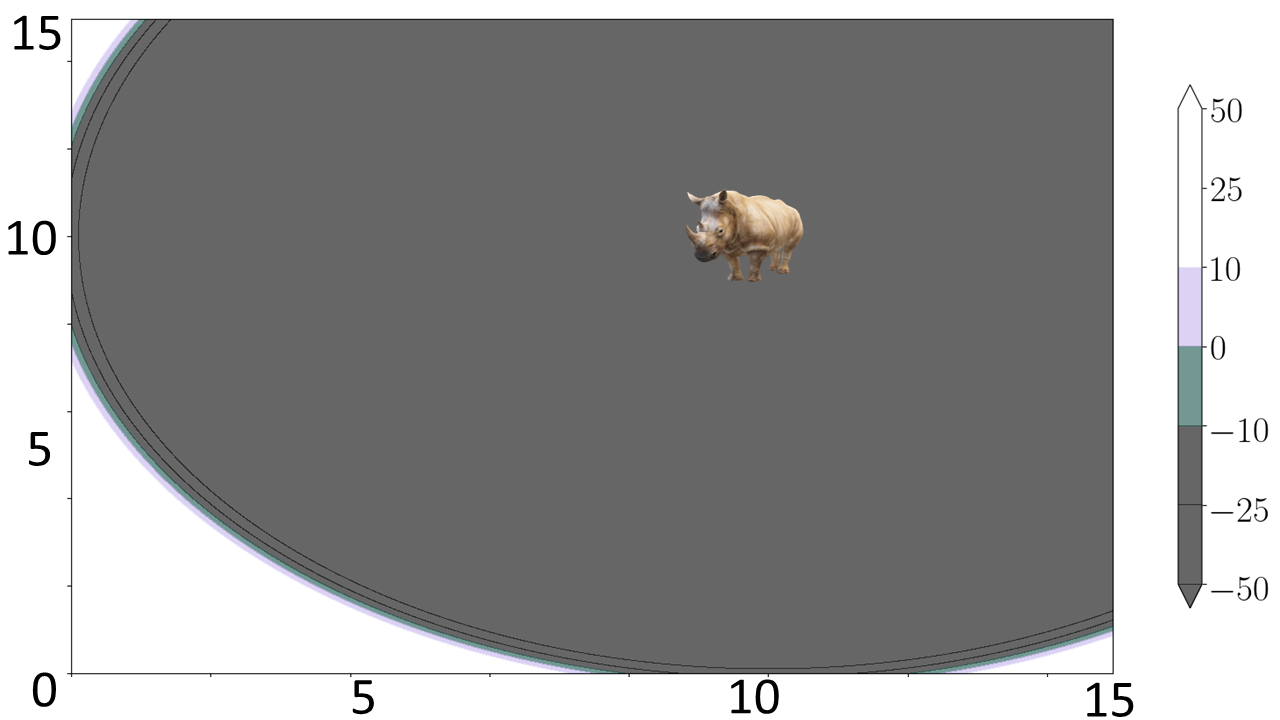}
 		\centering
 		\caption{Highly risk averse perception ($\lambda=100$ and $\rho=199$) with contour map depicting $\rho-\supscr{R}{cpt}$}
 		\label{fig:versatility_averse}
 	\end{subfigure}
	 ~
	 \begin{subfigure}[t]{0.33\linewidth}
	 	\includegraphics[width=\textwidth]{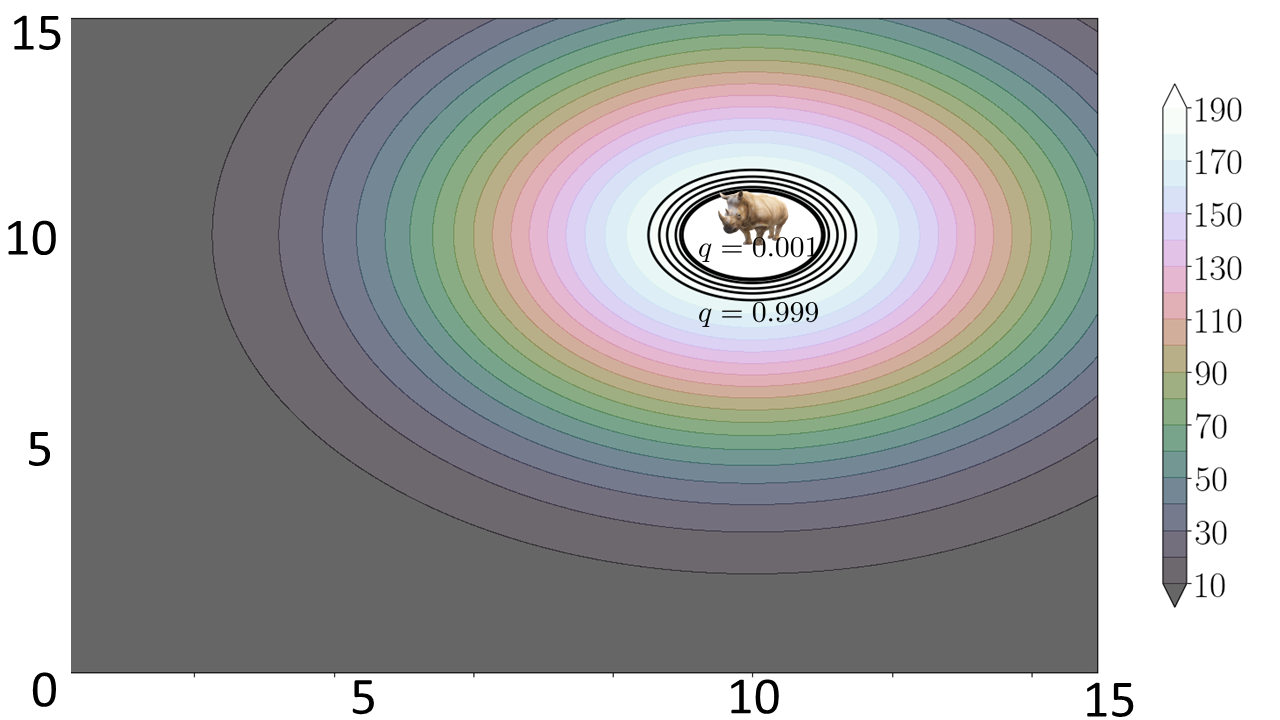}
	 	\centering
	 	\caption{Level sets of $R^q(\xi)=\rho$ by varying CVaR
	 		parameter $q \in \{0.001,0.1,0.4,0.8,0.95,0.999\}$}
	 	\label{fig:level_set_cvar}
	 \end{subfigure}
 	~
 	\begin{subfigure}[t]{0.33\linewidth}
 		\includegraphics[width=\textwidth]{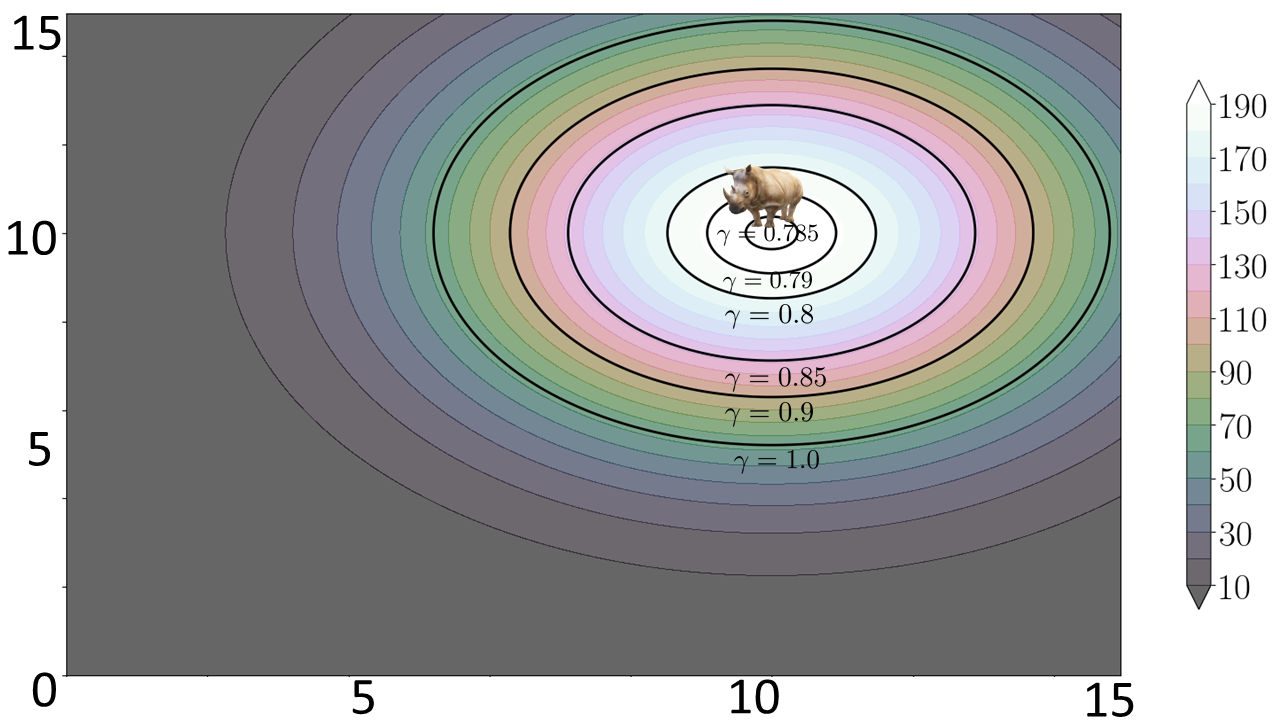}
 		\caption{Level sets of $\supscr{R}{cpt}(\xi)=\rho$ by varying risk
 			sensitivity parameter $\gamma \in
 			\{0.785,0.79,0.8,0.85,0.9,1.0\}$ with
 			$\alpha=0.74,\beta=1,\lambda=3.0$ }
 		\label{fig:level_set_gamma}
 	\end{subfigure}%
 	~
 	\begin{subfigure}[t]{0.33\linewidth}
 		\includegraphics[width=\textwidth]{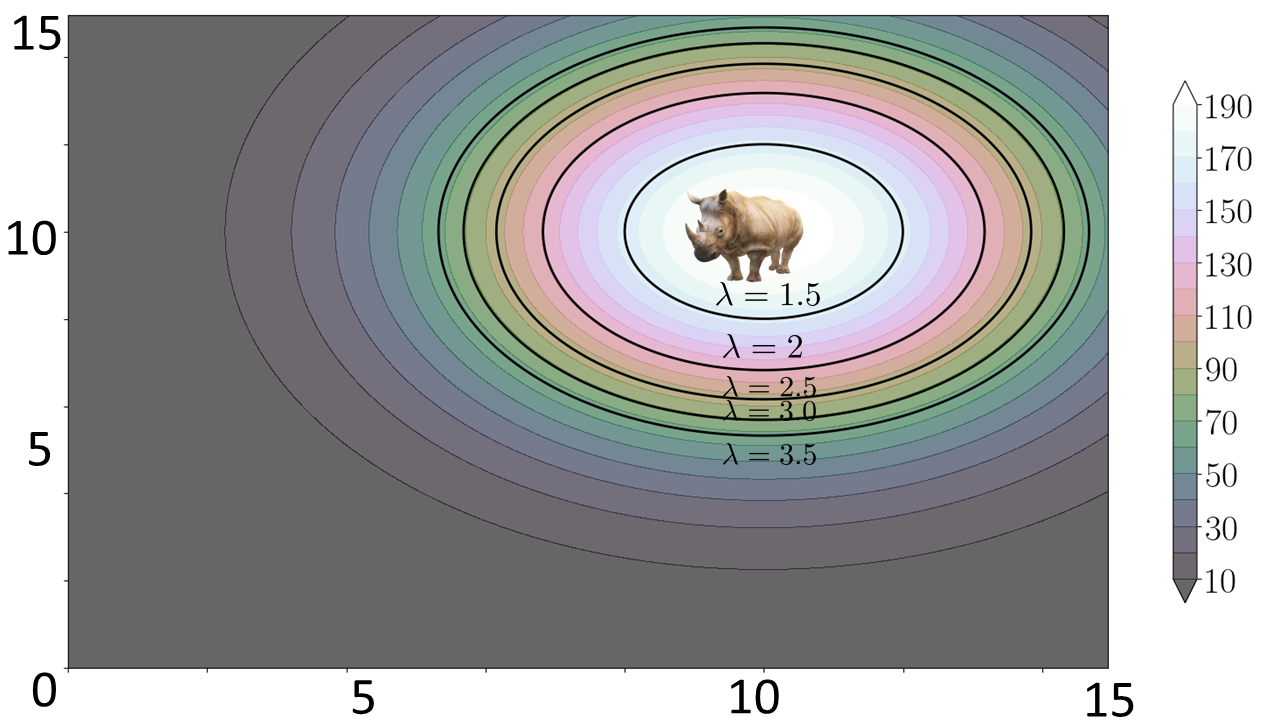}
 		\centering
 		\caption{Level sets of $\supscr{R}{cpt}(\xi)=\rho$ by varying risk
 			aversion parameter $\lambda \in \{1.5,2.0,2.5,3.0,3.5\}$
 			with $\alpha=0.74,\beta=1,\gamma=0.95$}
 		\label{fig:level_set_lambda}
 	\end{subfigure}
 	\caption[Caption]{(a)-(b) Cost maps $c_\mu$ and $c_\sigma$ over the space $\mathcal{X}=[0,15]^2$
          with obstacle's mean position at $y_\mu=\{10,10\}$. (c)-(d) Illustration of versatility of CPT through depicting scalar field $\rho - \supscr{R}{cpt}$ in previous setting. (e)-(f) Change in perceived safety sets
          $\subscr{\mathcal{X}}{safe}$ by varying CPT parameters $\gamma$ and $\lambda$. (g) Change in perceived safety sets
          $\subscr{\mathcal{X}}{safe}$ by varying CVaR parameter $q$  } 
 	\label{fig:cost_maps}
 \end{figure*}

%  \margin{use
%  	math mode for axis. I can barely see the ticks in the
%  	figures. Need modification.}

 \emph{Perceived safety visualization:} Under the previous setting, we
 provide visualizations of the costs and perceived risks, shown in
 Fig.~\ref{fig:cost_maps}. Fig.~\ref{fig:c_mu} and~\ref{fig:c_sigma}
 show the mean cost $c_\mu$ and standard deviation $c_\sigma$
 respectively, across $\mathcal{X}=[0,15]^2$ with obstacle's mean
 position at $y=(10,10)$ and $\subscr{c}{min}=1$ and
 $\subscr{c}{max}=200$.
% \margin{why do you use braces all of a
%   sudden to denote a vector? just use parenthesis (!) same comment
%   for the $y$ in the caption}.
 \emph{Versatility:} CPT's versatility is illustrated in
 Fig.~\ref{fig:versatility_insensitive} and
 Fig.~\ref{fig:versatility_averse} through contour maps of $h_R(\xi)=\rho
 - \supscr{R}{cpt}(x)$ 
 % \margin{make sure that we use the appropriate notation, I
 %   have changed things in terms of $h_R$ before, so do we have to use
 %   $h_R$? and also $h$ is a function of $\xi$ not of $x$, so please
 %   fix the argurments to make it consistent}
 across $\mathcal{X}$. Fig.~\ref{fig:versatility_insensitive} shows
 that despite risk threshold $\rho$ being very small ($\rho=27$) and
 close to $\subscr{c}{min}$, the entire space is perceived safe with
 positive $h$ values. In Fig.~\ref{fig:versatility_averse}, we observe
 the opposite, where a very high risk threshold value ($\rho=199$),
 close to $\subscr{c}{max}$ still makes almost the entire
 $\mathcal{X}$ unsafe with negative $h_R(\xi)$ values. This
 illustrates the versatility of CPT as an RPM in accordance with
 Proposition~\ref{prop:versatility}.
 
  \emph{Inclusiveness:}
 %  and show that CPT 
 %  can produce a better range of risk perception leading to a larger variety of 
 %  perceived safe and unsafe sets. 
 This concept is illustrated in
 Fig.~\ref{fig:level_set_cvar}--\ref{fig:level_set_lambda}.  The black
 lines indicate the level sets of $\supscr{R}{cpt}=\rho$ and
 $\supscr{R}{CV}=\rho$ evaluated by varying their respective
 parameters. From Fig.~\ref{fig:level_set_cvar} it is clear that
 variation in the level sets of CVaR is marginal compared to CPT
 (Fig.~\ref{fig:level_set_gamma} and \ref{fig:level_set_lambda}). The
 level set $\supscr{R}{ER}(\xi)=\rho$ is shown in
 Fig.~\ref{fig:level_set_cvar} as the inner most ellipse. We see that
 CPT is able to capture a more risk averse (larger) as well as more
 risk insensitive (smaller) perception than CVaR (and ER). This
 verifies the claims of Theorem~\ref{thm:inclusiveness} visually.
%\begin{figure*}
%	
%	\caption[Caption]{ }
%	\label{fig:level_sets}
%\end{figure*}

\emph{RPA-CBF controller:} We consider a single agent with
unicycle dynamics and a single obstacle whose dynamics evolve in the
space $\mathcal{X} \subset \real^2$. We use the costs defined in the
previous paragraph with $\overline{r}=0.5$ and $k_1=200,k_2=0.01$. The
agent starts at $x(0)=(5,2)^\top$ (green dot)
and its goal is $x^*=(10,10)^\top$ (motion up) while the obstacle
moves from $(13,13)^\top$ (rhino in red ellipse) to $(2,3)^\top$ (motion down).  If obstacle
and vehicle follow along straight paths, a collision would occur and
safety would be violated. To handle unicycle dynamics we use the
projected point method to control a virtual point $p \in
\real^2$, a distance $l$ along the direction of its heading. That is,
$p=x + l\vec{d}$, where $\vec{d}$ is the direction vector
corresponding to the agents heading $\phi$ and $x \in \real^2$ is the
planar coordinates of the agent. With this we get the reverse
transformation for the control inputs:
\begin{equation}
	\begin{bmatrix}
		u_1 \\ u_2	\end{bmatrix} =
	\begin{bmatrix}
		\cos(\phi) & \sin(\phi) \\
		\frac{-\sin(\phi)}{l} & \frac{\cos(\phi)}{l} 
	\end{bmatrix} u.
\label{eqn:point_offset_trans} 
\end{equation}
Where $u_1,u_2 \in \real$ are the linear and angular velocity inputs
of the unicycle model and $u \in \real^2$ is the optimized input
generated from~\eqref{eqn:clbf_opt} considering the $p$ dynamics
$\dot{p}=u$.  We use a standard proportional controller for $k(x)$ with a constant
$(0.6,0.6)$ 
%\margin{is this for $k$? or what? specify} to generate the
%nominal control signal $k(x)$. 
We note that one can always
appropriately tune the reference value $\rho$ by $l$ units to ensure
safety w.r.t. $x$. The results of varying $\lambda$, $\gamma$ and
$\kappa$ are shown in
Fig.~\ref{fig:results_traj_unicycle}. 
For all the settings, the agent will collide with the obstacle (red ellipse) if it
 follows the nominal path (black line) from applying controls $k(x)$, thus making it
  unsafe. By using the controller $u$ from \eqref{eqn:clbf_opt}, the agent is able to
   swerve away from the obstacle in time and still manage to reach the goal. We see 
   that from Fig.~\ref{fig:cbf_averse_uni} - Fig.~\ref{fig:cbf_cvar_uni}, the cbf $h_R$
    remains positive throughout the execution, thus indicating that perceived safety is 
    maintained irrespective of model and parameter choice. Next, we notice that by using 
    CPT-CBF controller $u$ (Fig.~\ref{fig:cbf_averse_uni} and 
    Fig.~\ref{fig:cbf_sensitive_uni}), the deviations from the nominal path correspondingly
     get more pronounced as the perceived risk increases (higher $\lambda$ and $\gamma$). 
     Whereas, for CVaR-CBF this deviation (Fig.~\ref{fig:cbf_cvar_uni}) is comparatively 
     minimal across its parameter spectrum. This is in accordance with our claims that 
     CVaR is less inclusive (Theorem~\ref{thm:inclusiveness}) and versatile 
     (Proposition~\ref{prop:versatility}) than CPT, causing only minor changes in 
     trajectories in comparison with CPT-based CBF controller.

 \begin{figure*}
	\begin{subfigure}[t]{0.30\linewidth}
		\includegraphics[width=\textwidth]{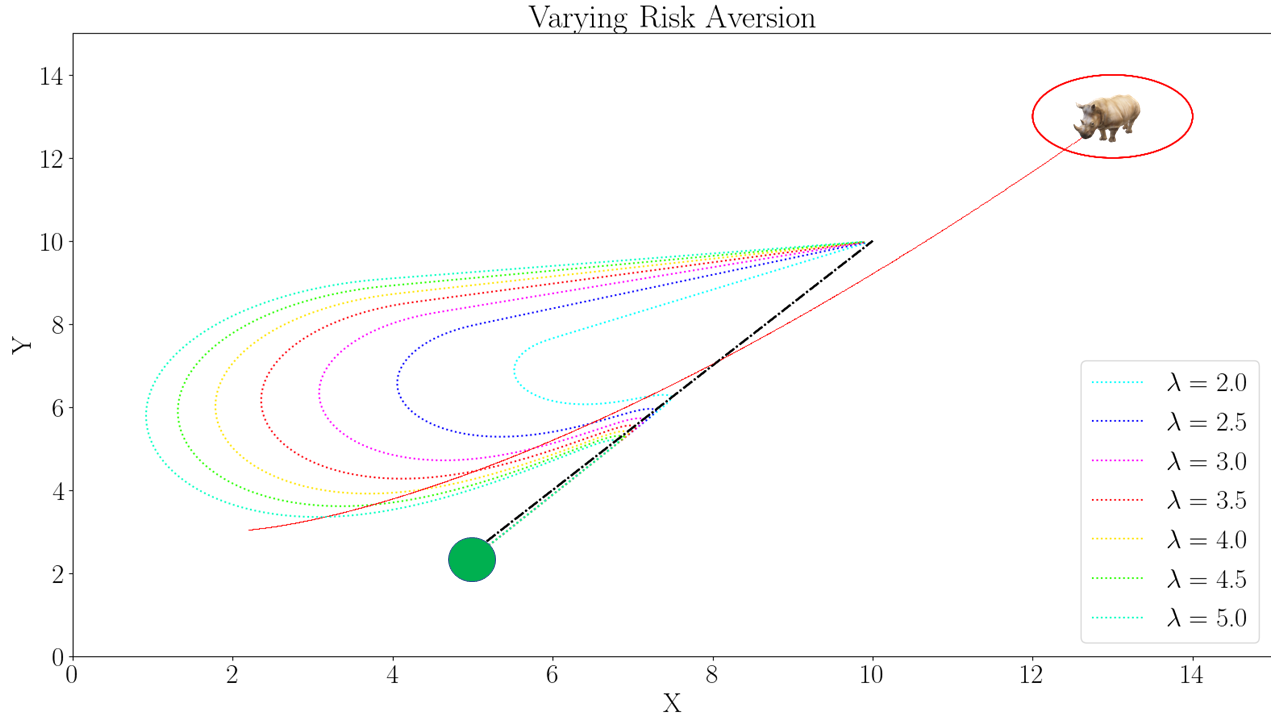}
		\caption{Change in paths due to change in risk
                  aversion $\lambda$ with $\gamma=0.88$ }
		\label{fig:path_averse_uni}
	\end{subfigure}%
	~
	\begin{subfigure}[t]{0.30\linewidth}
		\includegraphics[width=\textwidth]{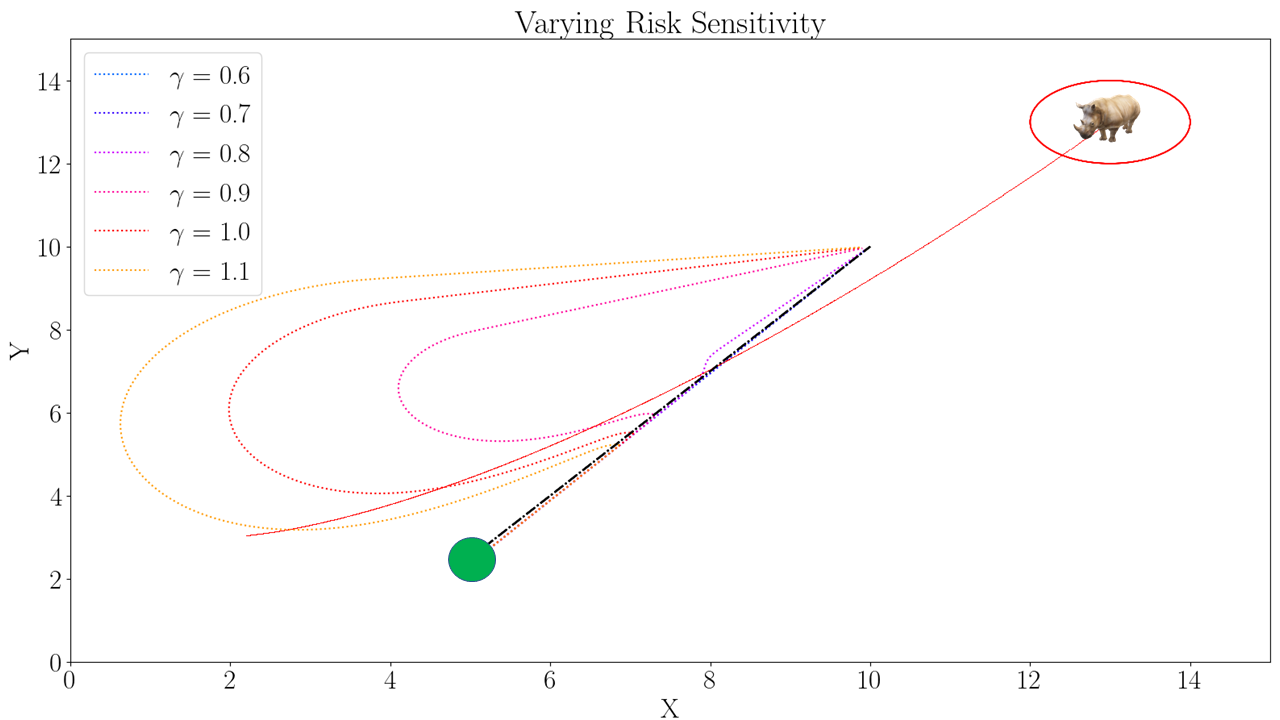}
		\centering
		\caption{Change in paths due to change in risk sensitivity $\gamma$ with $\lambda=2.25$}
		\label{fig:path_sensitive_uni}
	\end{subfigure}
	~
	\begin{subfigure}[t]{0.30\linewidth}
		\includegraphics[width=\textwidth]{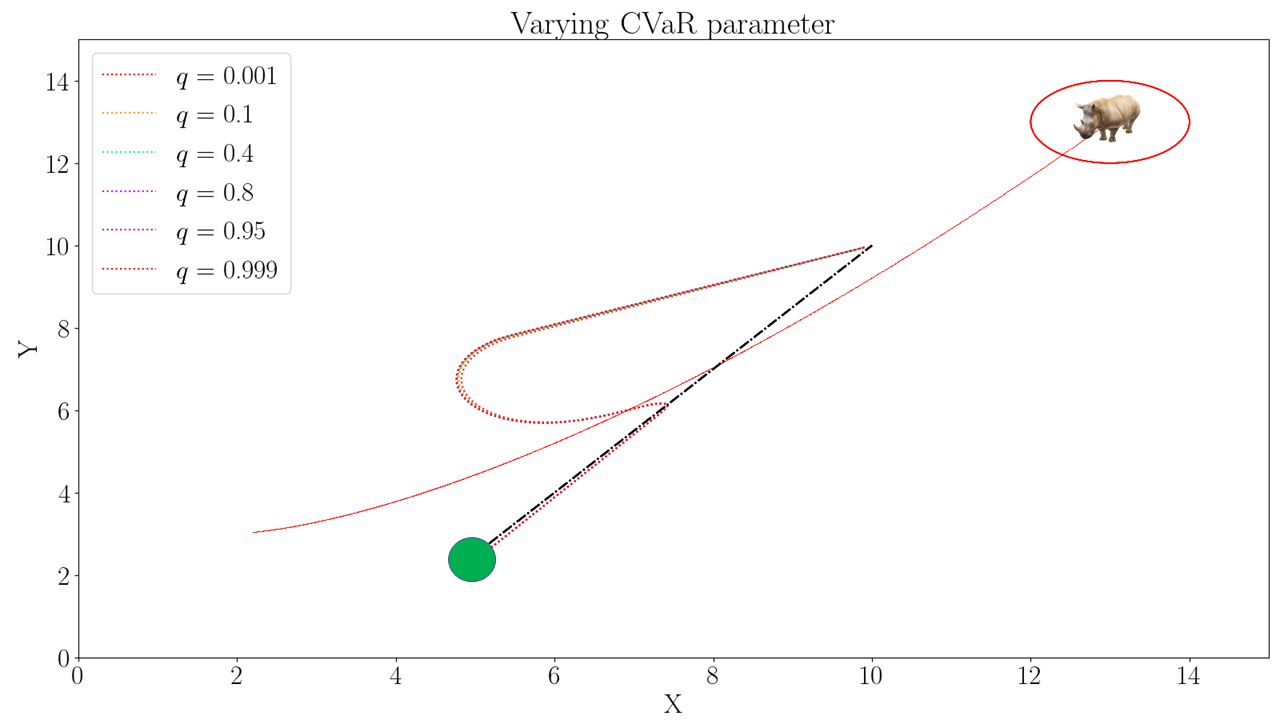}
		\centering
		\caption{Change in paths due to change in $\alpha$ with $\beta=1$}
		\label{fig:path_cvar_uni}
	\end{subfigure}
	
	\begin{subfigure}[t]{0.30\linewidth}
		\includegraphics[width=\textwidth]{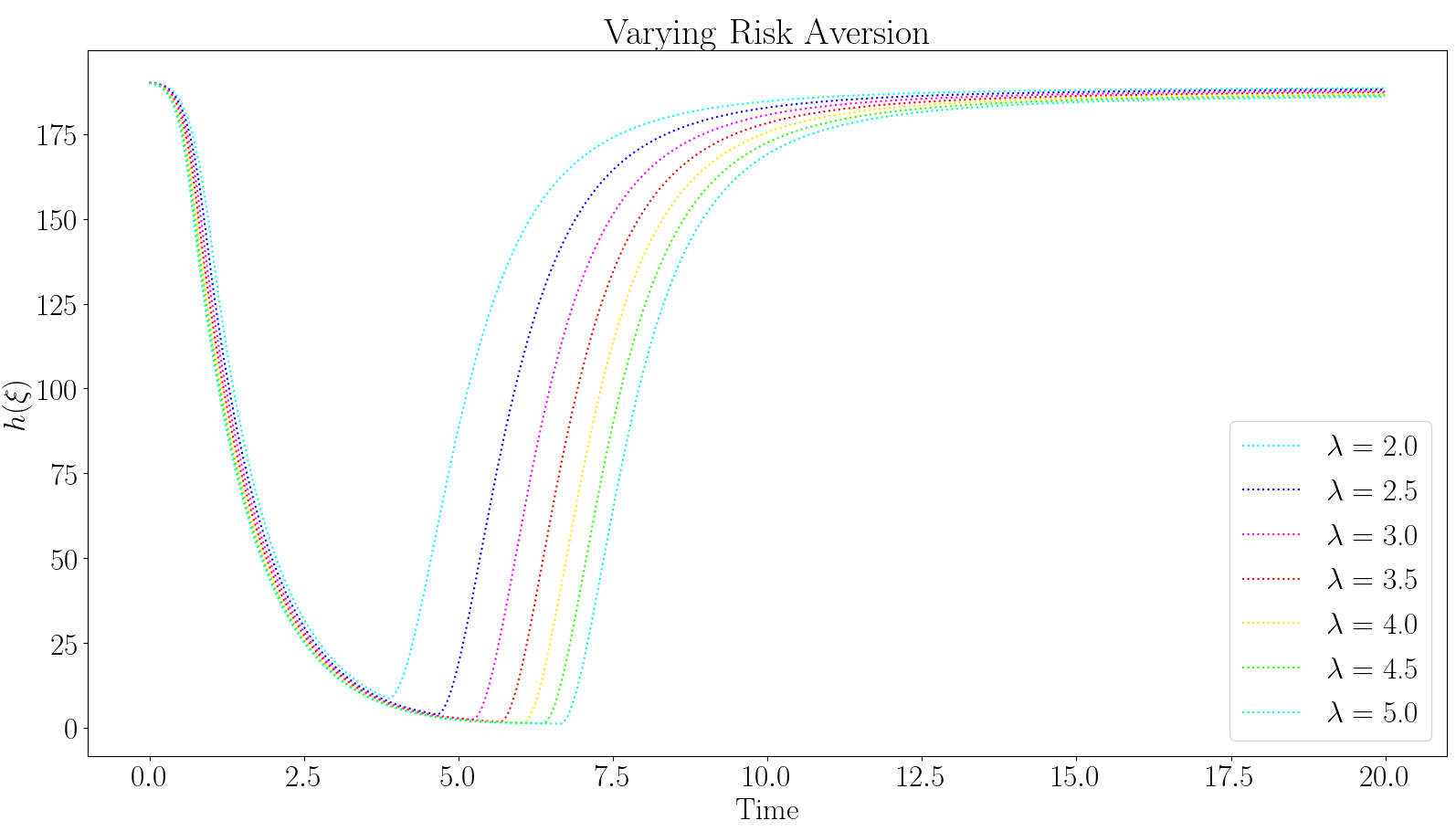}
		\centering
		\caption{Evolution of CBF function $h$ with varying $\lambda$}
		\label{fig:cbf_averse_uni}
	\end{subfigure}
	~
	\begin{subfigure}[t]{0.30\linewidth}
		\includegraphics[width=\textwidth]{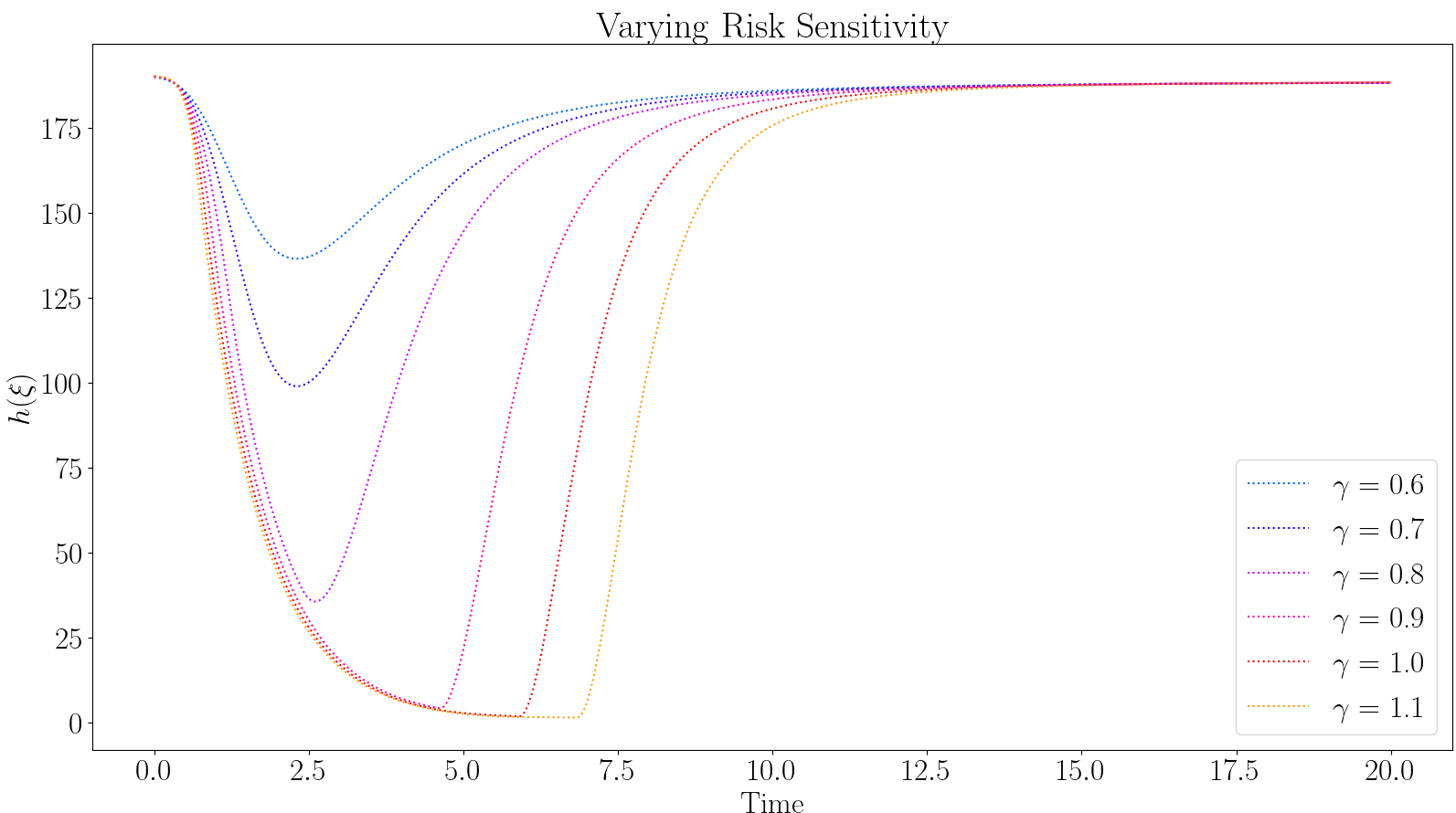}
		\centering
		\caption{Evolution of CBF function $h$ with varying $\gamma$}
		\label{fig:cbf_sensitive_uni}
	\end{subfigure}
	~
	\begin{subfigure}[t]{0.30\linewidth}
		\includegraphics[width=\textwidth]{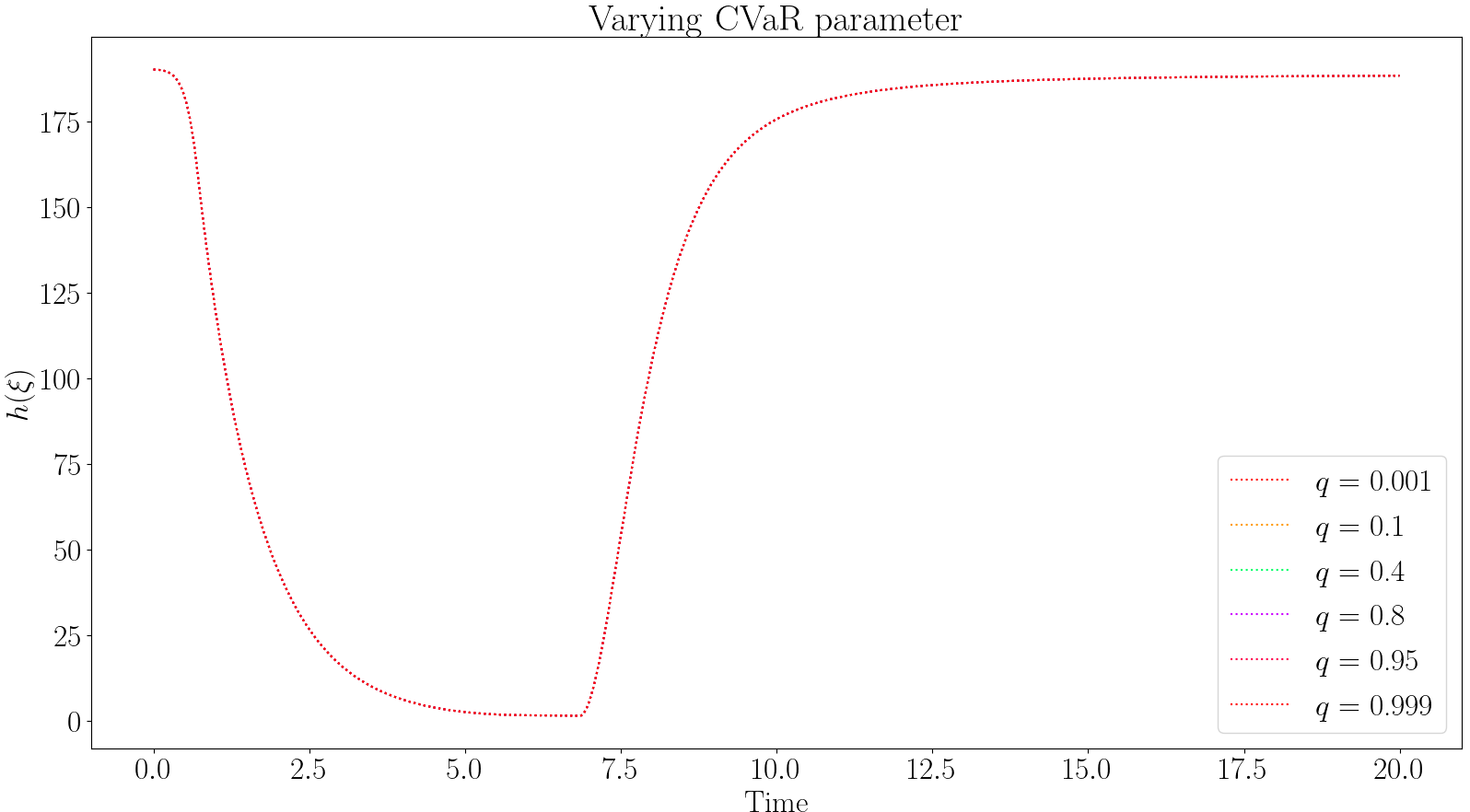}
		\centering
		\caption{Evolution of CBF function $h$ with varying $\kappa$ in CVaR}
		\label{fig:cbf_cvar_uni}
	\end{subfigure}
	\caption[Caption]{Path changes due to variation in risk aversion, risk sensitivity and
		CVaR parameter, with their corresponding CBF ($h$) evolution. The path of the obstacle (red) cross the nominal path of the agent (black) with the uncertainty circle indicated in red. }
	\label{fig:results_traj_unicycle}
\end{figure*}

From (Fig.~\ref{fig:results_traj_unicycle}), 
we see that the agent is able to reach the goal while maintaining 
$h \geq 0$ throughout, implying that perceived safety is maintained 
according to Definition~\ref{def:safety}. Furthermore, as before, 
we see that CPT is able to generate a wider range of paths by tuning 
the risk aversion and risk sensitivity parameter than CVaR, thus 
capturing a greater variety of risk perception, which follows the 
theoretical arguments from Theorem~\ref{thm:inclusiveness} and 
Proposition~\ref{prop:versatility}. 
We also see that the agent also reaches the goal 
owing to the inherent stability properties of the nominal 
controller $k(x)$. 

%%%%%%%% Multi agent setting %%%%%%%%%%%%%%%%%%%%%%%%%%%%%%%%
 Next, we consider an environment where there are three moving obstacles
  present and a single agent. We use the composition approach proposed 
  in~\cite{PG-JC-ME:21-tac} to construct the barrier function to handle 
  multiple obstacles. In this approach, the worst case (closest) 
  obstacle is dealt with first using the $\min$ operator on the barrier 
  functions generated by the corresponding obstacles. Here, the agent has 
  to go from $(-15,-15)$ to $(15,15)$, while the obstacles' starting and 
  goal points are respectively $(-17,0),(0,14),(10,-10)$ and $(17,0),
  (0,-14),(-10,10)$. Nominal controller $k(p)$ is generated with 
  proportional constant $[1.6,1.6]$. The uncertainty radius is 
  $\overline{r}=2.5$ and other cost constants are identical to the 
  single agent setting. The results of varying the risk aversion 
  $\lambda$, risk sensitivity $\gamma$ of CPT and $\kappa$ of CVaR 
  are shown in Fig. \ref{fig:results_traj_multi}.

 \begin{figure*}
	\begin{subfigure}[t]{0.30\linewidth}
 		\includegraphics[width=\textwidth]{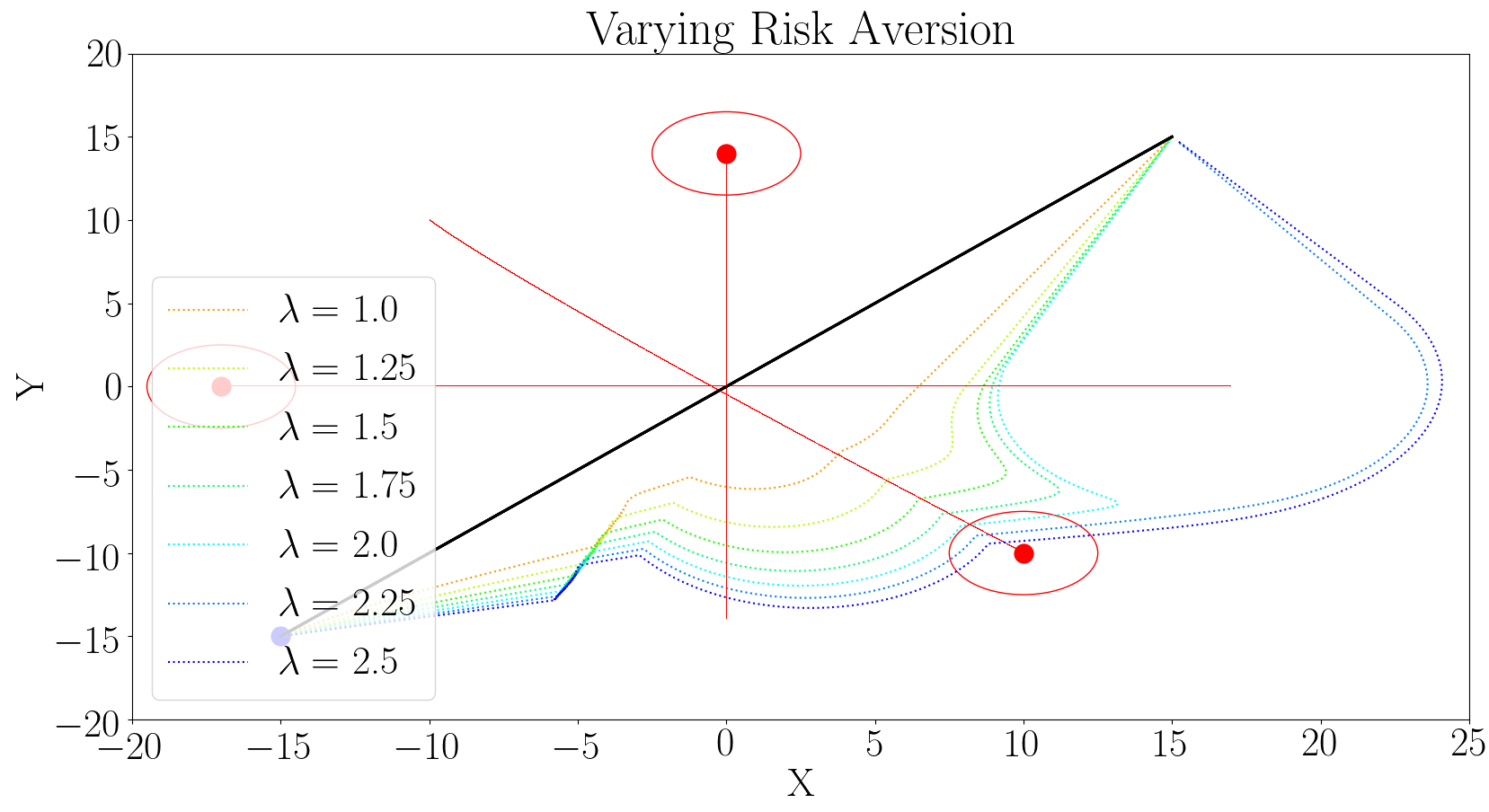}
 		\caption{Change in paths due to change in risk aversion $\lambda$ with $\gamma=0.88$ }
 		\label{fig:path_averse_multi}
 	\end{subfigure}%
	~
 	\begin{subfigure}[t]{0.30\linewidth}
 		\includegraphics[width=\textwidth]{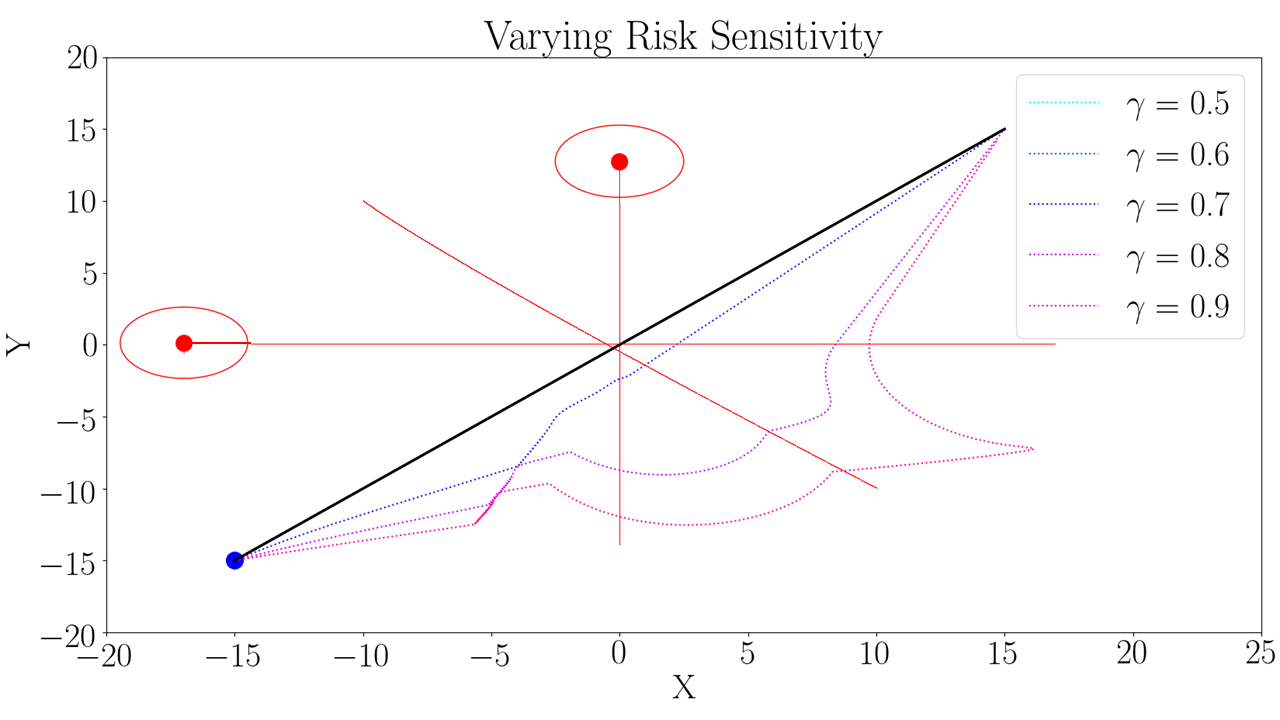}
 		\centering
		\caption{Change in paths due to change in risk sensitivity $\gamma$ with $\lambda=2.25$}
 		\label{fig:path_sensitive_multi}
 	\end{subfigure}
 	~
 	\begin{subfigure}[t]{0.30\linewidth}
 		\includegraphics[width=\textwidth]{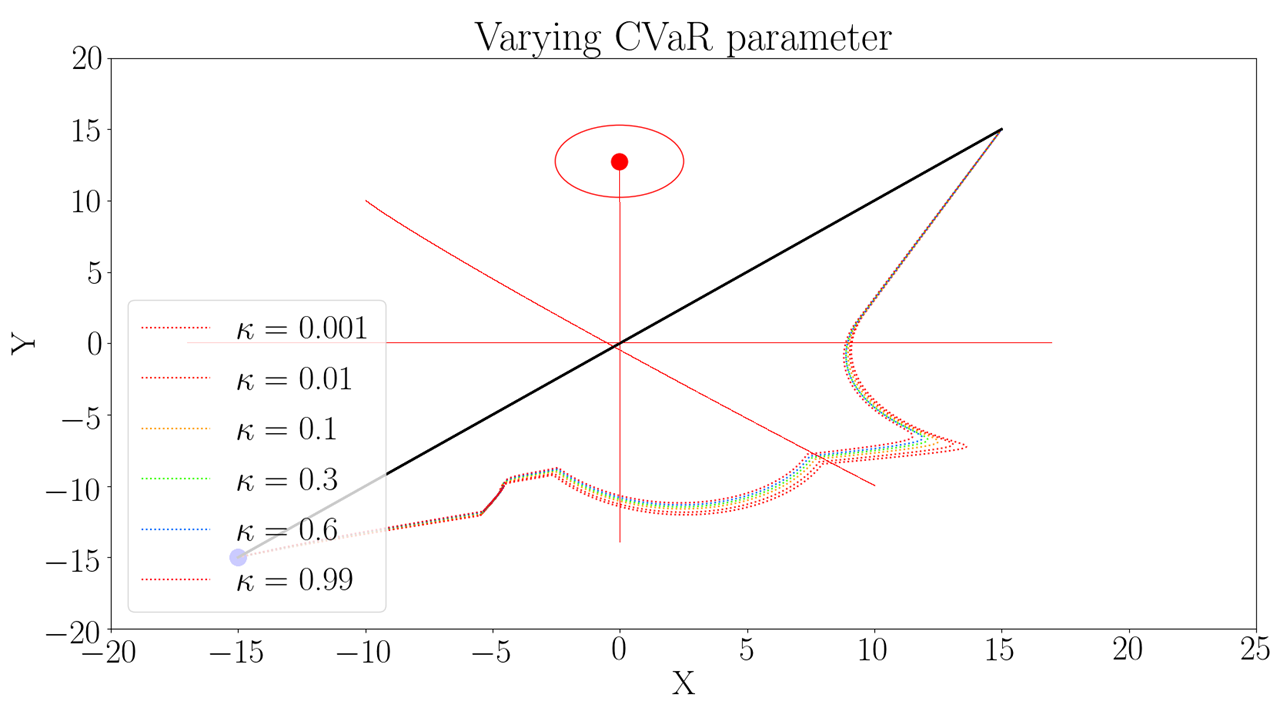}
 		\centering
 		\caption{Change in paths due to change in $\alpha$ with $\beta=1$}
 		\label{fig:path_cvar_multi}
 	\end{subfigure}
	
 	\begin{subfigure}[t]{0.30\linewidth}
 		\includegraphics[width=\textwidth]{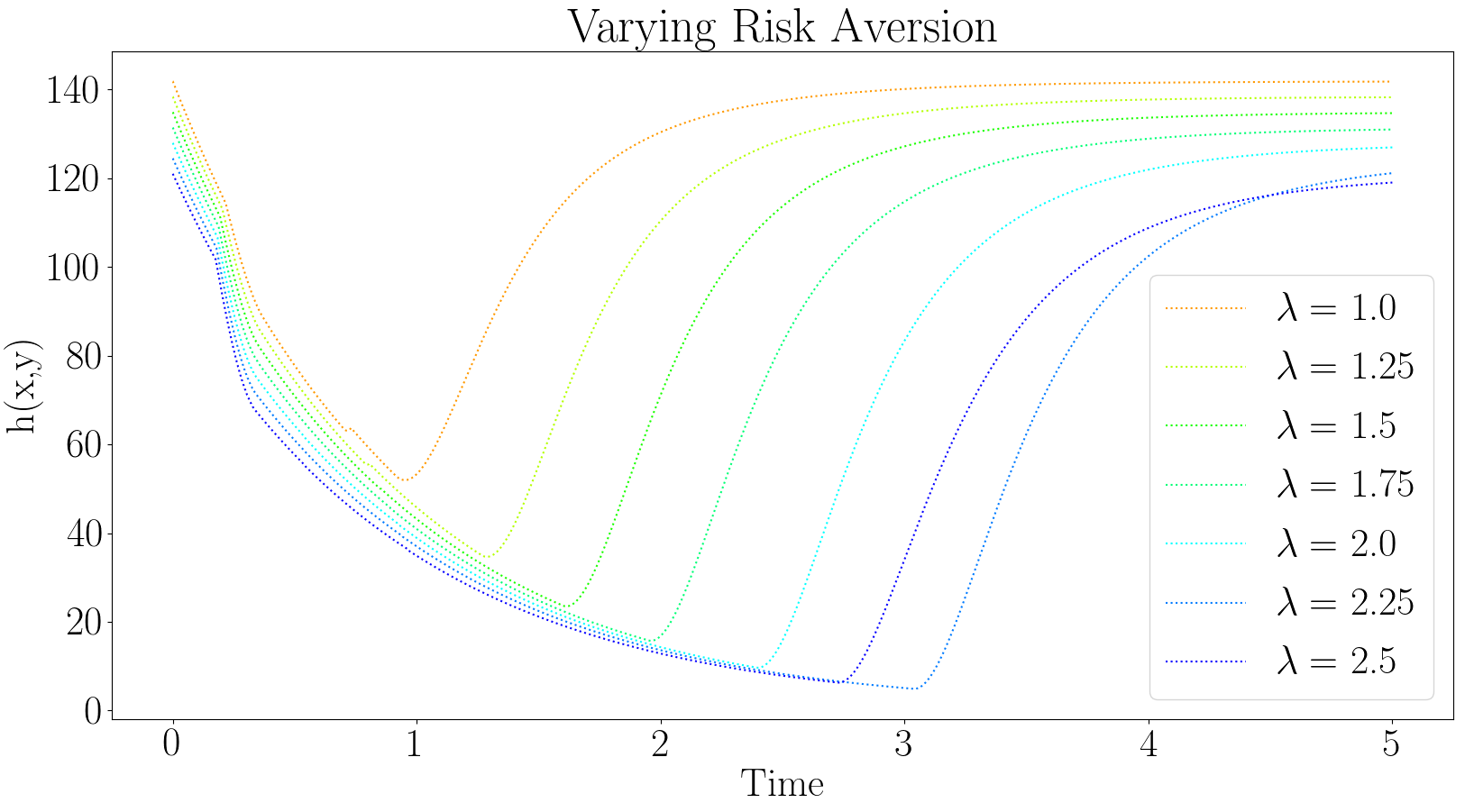}
 		\centering
 		\caption{Evolution of CBF function $h$ with varying $\lambda$}
 		\label{fig:cbf_averse_multi}
 	\end{subfigure}
 	~
	\begin{subfigure}[t]{0.30\linewidth}
 		\includegraphics[width=\textwidth]{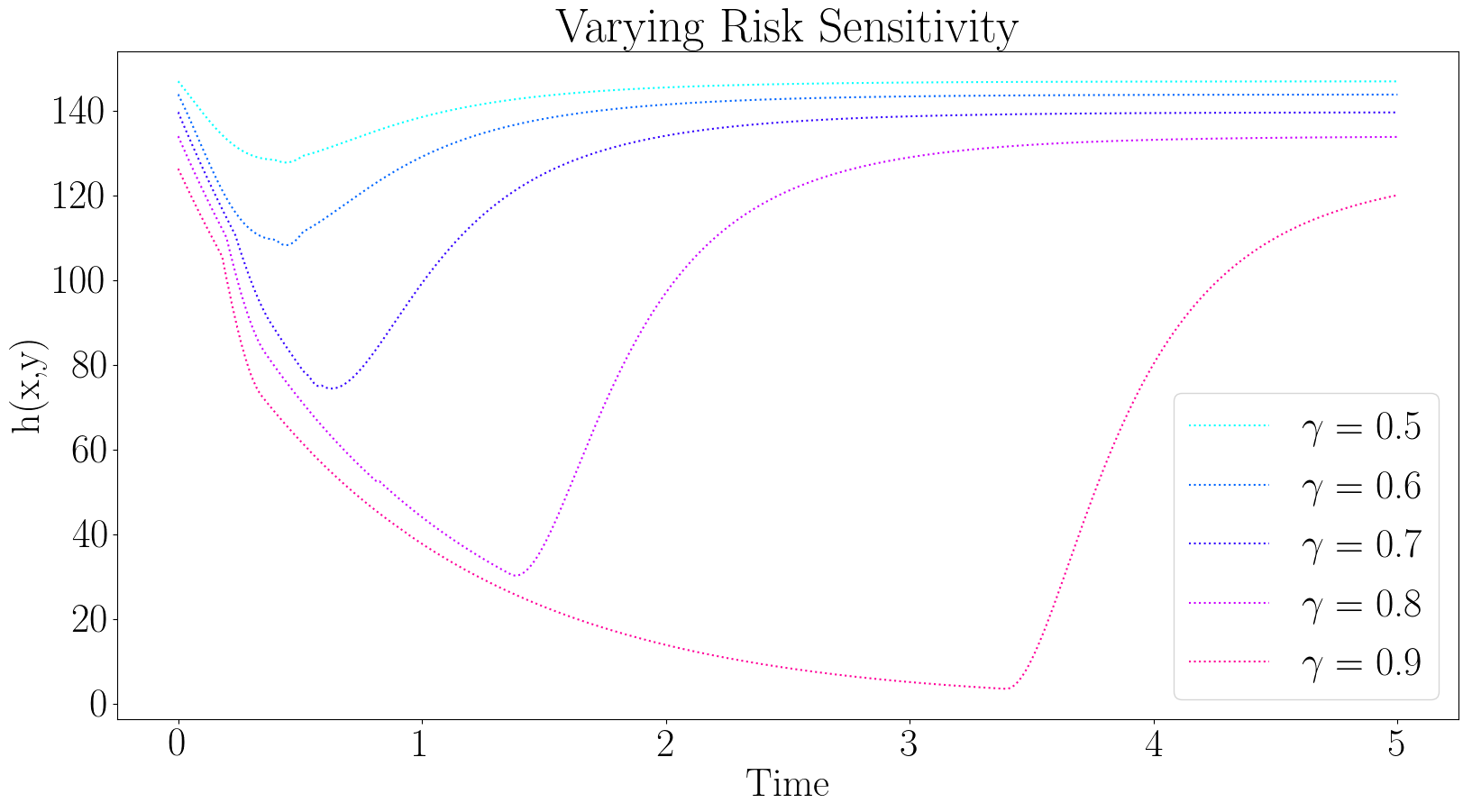}
 		\centering
 		\caption{Evolution of CBF function $h$ with varying $\gamma$}
 		\label{fig:cbf_sensitive_multi}
 	\end{subfigure}
 	~
 	\begin{subfigure}[t]{0.30\linewidth}
 		\includegraphics[width=\textwidth]{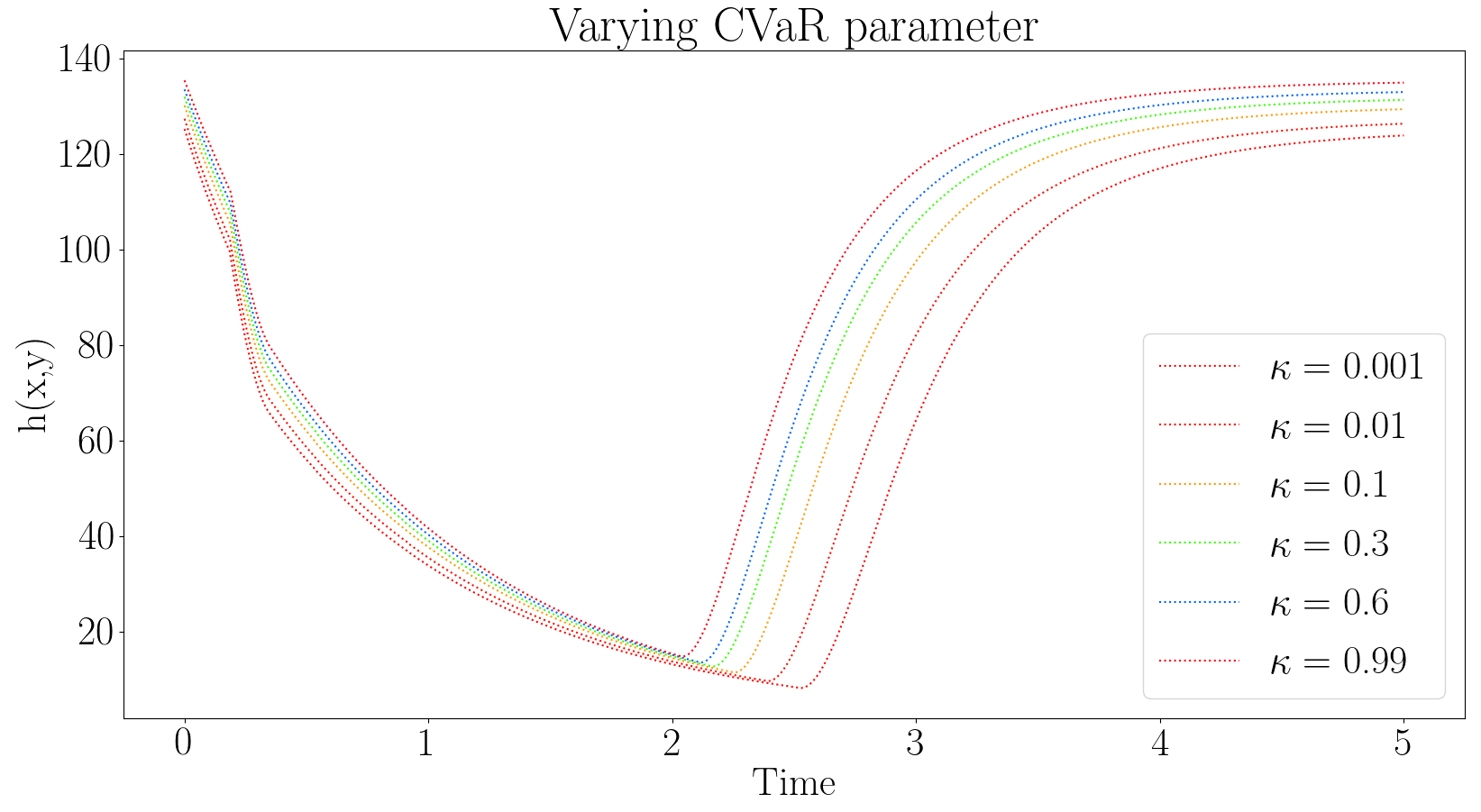}
 		\centering
 		\caption{Evolution of CBF function $h$ with varying $\kappa$ in CVaR}
 		\label{fig:cbf_cvar_multi}
 	\end{subfigure}
 	\caption[Caption]{Path changes due to variation in risk aversion, risk sensitivity and
 		CVaR parameter, with their corresponding CBF ($h$) evolution. The path of the obstacle (red) cross the nominal path of the agent (black) with the uncertainty circle indicated in red. }
 	\label{fig:results_traj_multi}
 \end{figure*}
 Similar to the previous case (Fig.~\ref{fig:results_traj_unicycle}), 
 we see that the agent is able to reach the goal while maintaining 
 $h \geq 0$ throughout, implying that perceived safety is maintained 
 according to Definition~\ref{def:safety}. Furthermore, as before, 
 we see that CPT is able to generate a wider range of paths by tuning 
 the risk aversion and risk sensitivity parameter than CVaR, thus 
 capturing a greater variety of risk perception, which follows the 
 theoretical arguments from Theorem~\ref{thm:inclusiveness} and 
 Proposition~\ref{prop:versatility}. 
 We also see that in both the cases the agent also reaches the goal 
 owing to the inherent stability properties of the nominal 
 controller $k(x)$.

%----------------COnclusions-----------------
\section{Conclusion and future work}
\label{sec:conclussions}

In this work, we have proposed a novel integration of CPT 
(a non-rational decision making model) 
into a safety-critical control scheme, 
to generate risk-perception-aware (RPA) controls 
(according to a DM's risk profile) in an environment 
embedded with uncertain costs. Thus, opening new avenues to 
incorporate behavioral decision theory into safety-critical controls.
Future directions include the design of learning frameworks to determine 
the risk profile of an observed agent and handling unknown obstacle 
dynamics.  
\bibliographystyle{IEEEtran}
\bibliography{alias,SM,SMD-add,JC}

\begin{thebibliography}{10}
\providecommand{\url}[1]{#1}
\csname url@samestyle\endcsname
\providecommand{\newblock}{\relax}
\providecommand{\bibinfo}[2]{#2}
\providecommand{\BIBentrySTDinterwordspacing}{\spaceskip=0pt\relax}
\providecommand{\BIBentryALTinterwordstretchfactor}{4}
\providecommand{\BIBentryALTinterwordspacing}{\spaceskip=\fontdimen2\font plus
\BIBentryALTinterwordstretchfactor\fontdimen3\font minus
  \fontdimen4\font\relax}
\providecommand{\BIBforeignlanguage}[2]{{%
\expandafter\ifx\csname l@#1\endcsname\relax
\typeout{** WARNING: IEEEtran.bst: No hyphenation pattern has been}%
\typeout{** loaded for the language `#1'. Using the pattern for}%
\typeout{** the default language instead.}%
\else
\language=\csname l@#1\endcsname
\fi
#2}}
\providecommand{\BIBdecl}{\relax}
\BIBdecl

\bibitem{SSS:70}
S.~S. Stevens, ``Neural events and the psychophysical law,'' \emph{Science},
  vol. 170, no. 3962, pp. 1043--1050, 1970.

\bibitem{AT-DK:92}
A.~Tversky and D.~Kahneman, ``Advances in prospect theory: Cumulative
  representation of uncertainty,'' \emph{Journal of Risk and Uncertainty},
  vol.~5, no.~4, pp. 297--323, 1992.

\bibitem{OK:90}
O.~Khatib, \emph{Real-{T}ime {O}bstacle {A}voidance for {M}anipulators and
  {M}obile {R}obots}.\hskip 1em plus 0.5em minus 0.4em\relax Springer New York,
  1990, pp. 396--404.

\bibitem{SP-AJ:04}
S.~Prajna and A.~Jadbabaie, ``{S}afety {V}erification of {H}ybrid {S}ystems
  using {B}arrier {C}ertificates,'' in \emph{Hybrid Systems: Computation and
  Control}, 2004, pp. 477--492.

\bibitem{ADA-SC-ME-GN-KS-PT:19}
A.~D. {Ames}, S.~{Coogan}, M.~{Egerstedt}, G.~{Notomista}, K.~{Sreenath}, and
  P.~{Tabuada}, ``Control barrier functions: Theory and applications,'' in
  \emph{{E}uropean {C}ontrol {C}onference}, 2019, pp. 3420--3431.

\bibitem{PO-JC:19-cdc}
P.~Ong and J.~Cort\'es, ``Universal formula for smooth safe stabilization,'' in
  \emph{{IEEE} Int. Conf. on Decision and Control}, Nice, France, Dec. 2019,
  pp. 2373--2378.

\bibitem{BTL-JJS-JPH:21}
B.~T. {Lopez}, J.~J.~E. {Slotine}, and J.~P. {How}, ``Robust adaptive control
  barrier functions: An adaptive and data-driven approach to safety,''
  \emph{IEEE Control Systems Letters}, vol.~5, no.~3, pp. 1031--1036, 2021.

\bibitem{YC-HP-JG:18}
Y.~{Chen}, H.~{Peng}, and J.~{Grizzle}, ``Obstacle avoidance for low-speed
  autonomous vehicles with barrier function,'' \emph{IEEE Transactions on
  Control Systems Technology}, vol.~26, no.~1, pp. 194--206, 2018.

\bibitem{PG-JC-ME:21-tac}
P.~Glotfelter, J.~Cort\'es, and M.~Egerstedt, ``Nonsmooth approach to
  controller synthesis for {B}oolean specifications,'' \emph{IEEE Transactions
  on Automatic Control}, vol.~66, no.~11, 2021, to appear.

\bibitem{SP-AJ-GJP:07}
S.~{Prajna}, A.~{Jadbabaie}, and G.~J. {Pappas}, ``A framework for worst-case
  and stochastic safety verification using barrier certificates,'' \emph{IEEE
  Transactions on Automatic Control}, vol.~52, no.~8, pp. 1415--1428, 2007.

\bibitem{AC:19}
A.~Clark, ``Control {B}arrier {F}unctions for {C}omplete and {I}ncomplete
  {I}nformation {S}tochastic {S}ystems,'' in \emph{{A}merican {C}ontrol
  {C}onference}, 2019, pp. 2928--2935.

\bibitem{MJK-VD-MF-NA:20}
M.~J. Khojasteh, V.~Dhiman, M.~Franceschetti, and N.~Atanasov, ``Probabilistic
  safety constraints for learned high relative degree system dynamics,'' in
  \emph{Proceedings of Learning for Dynamics and Control}, 2020, pp. 781--792.

\bibitem{SS-IY:18}
S.~{Samuelson} and I.~{Yang}, ``Safety-aware optimal control of stochastic
  systems using conditional value-at-risk,'' in \emph{{A}merican {C}ontrol
  {C}onference}, 2018, pp. 6285--6290.

\bibitem{MA-XX-ADA:21}
M.~Ahmadi, X.~Xiong, and A.~D. Ames, ``Risk-averse control via cvar barrier
  functions: Application to bipedal robot locomotion,'' \emph{IEEE Control
  Systems Letters}, vol.~6, no.~1, pp. 878--883, 2021.

\bibitem{AS-SM:21-ral}
A.~Suresh and S.~Mart{\'\i}nez, ``Planning under non-rational perception of
  uncertain spatial costs,'' \emph{IEEE Robotics and Automation Letters},
  vol.~6, no.~2, pp. 4133--4140, 2021.

\bibitem{SG-EF-MBA:10}
S.~Gao, E.~Frejinger, and M.~Ben-Akiva, ``Adaptive route choices in risky
  traffic networks: A prospect theory approach,'' \emph{Transportation Research
  Part C: Emerging Technologies}, vol.~18, no.~5, pp. 727--740, 2010.

\bibitem{ARH-SS:19}
A.~R. Hota and S.~Sundaram, ``Game-{T}heoretic {P}rotection against {N}etworked
  sis {E}pidemics by {H}uman {D}ecision-{M}akers,'' \emph{IFAC Papers Online},
  vol.~51, no.~34, pp. 145--150, 2019.

\bibitem{SD:16}
S.~Dhami, \emph{The Foundations of Behavioral Economic Analysis}.\hskip 1em
  plus 0.5em minus 0.4em\relax Oxford University press, 2016.

\bibitem{RTR-SU:00}
R.~T. Rockafellar and S.~Uryasev, ``Optimization of conditional
  value-at-risk,'' \emph{Journal of risk}, vol.~2, no.~1, pp. 21--42, 2000.

\bibitem{DP:98}
D.Prelec, ``The probability weighing function,'' \emph{Econometrica}, vol.~66,
  no.~3, pp. 497--527, 1998.

\end{thebibliography}

\end{document}